\documentclass[a4paper,onecolumn,10pt,accepted=2025-10-30]{quantumarticle}
\pdfoutput=1
\usepackage[english]{babel}
\usepackage[utf8]{inputenc}
\usepackage[colorinlistoftodos, color=green!40, prependcaption]{todonotes}
\usepackage{amsthm}
\usepackage{amssymb}
\usepackage{comment}
\usepackage{mathtools}
\DeclarePairedDelimiter{\ceil}{\lceil}{\rceil}
\usepackage{physics}
\usepackage{chngcntr}
\usepackage{xcolor}
\usepackage{makecell}
\usepackage{graphicx}
\usepackage[export]{adjustbox}
\usepackage{graphbox}
\usepackage{placeins}
\usepackage[T1]{fontenc}
\usepackage{lipsum}
\usepackage{centernot}
\usepackage{csquotes}
\mathtoolsset{showonlyrefs}
\usepackage[pdftex, pdftitle={Article}, pdfauthor={Author},colorlinks = true, linkcolor = blue, urlcolor  = blue, citecolor = red]{hyperref} 

\makeatletter
\@addtoreset{equation}{myequation}
\makeatother

\makeatletter
\@addtoreset{theorem}{mytheorem}
\makeatother

\makeatletter
\@addtoreset{section}{mysection}
\makeatother

\usepackage[title]{appendix}

\usepackage{amsfonts}
\usepackage{amsmath}
\usepackage{stmaryrd}
\usepackage{braket}
\usepackage{dsfont}
\usepackage{bbold}
\usepackage{relsize}
\usepackage{float}

\usepackage[font=scriptsize]{caption}
\let\ACMmaketitle=\maketitle
\renewcommand{\maketitle}{\begingroup\let\footnote=\thanks \ACMmaketitle\endgroup}

\newcommand{\satvik}[1]{{\color{cyan} #1}}

\newcommand{\Hil}{\mathcal{H}}

\newcommand{\cH}{\mathcal{H}}

\newcommand{\cC}{\mathcal{C}}
\newcommand{\cD}{\mathcal{D}}
\newcommand{\cL}{\mathcal{L}}
\newcommand{\cN}{\mathcal{N}}
\newcommand{\cX}{\mathcal{X}}
\newcommand{\cY}{\mathcal{Y}}
\newcommand{\iden}{\mathbb{1}}
\renewcommand{\epsilon}{\varepsilon}
\renewcommand{\phi}{\varphi}
\newcommand{\overbar}[1]{\mkern 1.5mu\overline{\mkern-1.5mu#1\mkern-1.5mu}\mkern 1.5mu}

\newcommand{\M}[1]{\mathbb{M}_{#1}(\mathbb{C})}
\newcommand{\St}[1]{\mathbb{S}_{#1}(\mathbb{C})}
\newcommand{\C}[1]{\mathbb{C}^{#1}}

\newcommand{\B}[1]{\mathcal{L}({#1})}
\newcommand{\State}[1]{\mathcal{D}({#1})}

\newcommand{\ps}{P_\star}
\newcommand{\rhos}{\rho_\star}
\newcommand{\tPhin}{{\tilde{\Phi}}_n}
\newcommand{\id}{{\rm{id}}}
\newcommand{\Ss}{{S}_\star}
\newcommand{\supp}{{\rm{supp }}}
\newtheorem{theorem}{Theorem}[section]
\newtheorem{definition}[theorem]{Definition}
\newtheorem*{definition*}{Definition}
\newtheorem{proposition}[theorem]{Proposition}
\newtheorem{corollary}[theorem]{Corollary}
\newtheorem{lemma}[theorem]{Lemma}
\newtheorem{remark}[theorem]{Remark}

\newtheorem*{conjecture*}{Conjecture}

\newcommand\vertarrowbox[3][6ex]{%
  \begin{array}[t]{@{}c@{}} #2 \\
  \left\uparrow\vcenter{\hrule height #1}\right.\kern-\nulldelimiterspace\\
  \makebox[0pt]{\scriptsize#3}
  \end{array}%
}

\theoremstyle{definition}
\newtheorem{example}[theorem]{Example}
\newtheorem{rem}{Remark}

\begin{document}

\title{Zero-error communication under discrete-time Markovian dynamics}

\author{Satvik Singh}
\affiliation{Department of Mathematics, Technical University of Munich}
\affiliation{DAMTP, Centre for Mathematical Sciences, University of Cambridge}

\author{Mizanur Rahaman}
\affiliation{UNIV LYON, INRIA, ENS LYON, UCBL, LIP, F-69342, LYON CEDEX 07, FRANCE}
\affiliation{Wallenberg Centre for Quantum Technology, Chalmers University of Technology}
\affiliation{Department of Mathematical Sciences, Chalmers University of Technology}

\author{Nilanjana Datta}
\affiliation{DAMTP, Centre for Mathematical Sciences, University of Cambridge}

\begin{abstract}
Consider an open quantum system with (discrete-time) Markovian dynamics. Our task is to store information in the system in such a way that it can be retrieved perfectly, even after the system is left to evolve for an arbitrarily long time. We show that this is impossible for classical (resp. quantum) information precisely when the dynamics is mixing (resp. asymptotically entanglement breaking). Furthermore, we provide tight universal upper bounds on the minimum time after which any such dynamics `scrambles' the encoded information beyond the point of perfect retrieval. On the other hand, for dynamics that are not of this kind, we show that information must be encoded inside the peripheral space associated with the dynamics in order for it to be perfectly recoverable at any time in the future. This allows us to derive explicit formulas for the maximum amount of information that can be protected from noise in terms of the structure of the peripheral space of the dynamics.
\end{abstract}

\maketitle

\section{Introduction}

Systems of relevance in quantum information-processing tasks are typically {\em{open}}, i.e.~they have unavoidable interactions with their surroundings. The external system modelling the surrounding of the original system is usually called its environment or a bath. An interesting scenario, which is amenable to rigorous analysis, is one in which the interaction between the system and the bath is assumed to be weak. In this so-called {\em{weak-coupling limit}}, the decay times of correlation functions of the bath are much
shorter than the typical time scale over which the state of the system changes
significantly. In other words, the bath `forgets' about its interaction with the system and returns to its steady state quickly relative to the speed at which the system evolves. Since in subsequent interactions, the bath does not remember the details of the previous interaction, the dynamics of the system becomes Markovian. Mathematically, the reduced dynamics of the system can be modelled by a (discrete- or continuous-time) quantum Markov semigroup~\cite{Fagnola1999semigroup, Alicki2002QDS, Alicki2007QDS, Breuer2007open, Wolf2012Qtour}.

In this paper, we consider a finite-dimensional open quantum system comprised of a number of qubits, say $l\in \mathbb{N}$, which undergoes a Markovian evolution. Our task is to store information in the system in such a way that it can be recovered {\em{perfectly}} without any error, even after the system is left to evolve for an arbitrarily long time. We can think of the system as a quantum memory in which we wish to store information so that it can be perfectly retrieved in the future. We focus on the discrete-time scenario, where the evolution of the system is given by a discrete-time quantum Markov semigroup (dQMS). If $\Hil \simeq (\C{2})^{\otimes l}$ is the Hilbert space of the system, then any dQMS is of the form $\{\Phi^n \}_{n\in \mathbb N}$, where $\Phi:\cL (\Hil)\to \B{\Hil}$ is a quantum channel (i.e.~a linear, completely positive, and trace-preserving map between linear operators acting on $\Hil$) and 
\begin{equation}
    \Phi^n := \underbrace{\Phi \circ \Phi \circ \ldots \circ \Phi}_{n\,  \text{times}}
\end{equation}
is the $n$-fold composition of the channel with itself. The nomenclature arises from the fact that compositions of the channel clearly satisfy the semigroup property: $\Phi^{n+m}=\Phi^n \circ \Phi^m$, for all $n,m \in {\mathbb{N}}$. The task described above can be thought of as communicating information across time, i.e.~through channels $\Phi^n$, where $n\in \mathbb{N}$ plays the role of the time parameter. In Shannon theory \cite{Shannon1948info, Shannon1956zero, Wilde2013quantum}, such communication problems are traditionally studied in the so-called \emph{asymptotic} and \emph{memoryless} setting, where one analyzes the optimal \emph{rates} of information transmission via many independent and identical (IID) uses, say $l$, of a given channel in the limit of $l\to \infty$. For the data storage problem that we consider, this framework is unsuitable for two reasons. Firstly, the noise in the memory might act in a correlated fashion across some of the qubits, so that the channel $\Phi$ might not act independently and identically on all $l$ qubits inside the memory. Secondly, since current quantum technologies can only coherently manipulate a few hundred qubits at most \cite{Preskill2018nisq}, it is pertinent to analyze storage capacities of memory devices with a small number of qubits ($l\sim 100$), thus making the asymptotic $l\to \infty$ limit rather unrealistic. These concerns are addressed by the framework of \emph{one-shot} information theory \cite{khatri2024principles}, where the focus is on determining how much information can be transmitted via a \emph{single} use of a given channel with some allowed transmission error. We now introduce this framework in the zero-error setting. 

%Such tasks form the core of zero-error communication theory \cite{Shannon1956zero, Medeiros2006zero, Beigi2007zero, Chen2010zero, Duan2013noncomm}. However, both in the classical and quantum settings, the focus till now has been on information transmission through \emph{parallel} uses (i.e.~tensor products) of a channel \cite{Cover2005info, Hayashi2017quantum, watrous2018theory, Wilde2013quantum}. In fact, originally, capacities of channels were evaluated in the so-called {\em{asymptotic, memoryless}} setting in which it was assumed that the channel was available for an infinite number of (parallel) uses and that there was no correlation (or memory) between successive uses of the channel. For a channel $\Phi$, this was modelled by considering information transmission through $\Phi^{\otimes n}$ and then taking the limit $n \to \infty$. In this paper, we instead focus on how the capacities behave under \emph{sequential} concatenations of channels, thereby studying how information propagates over time under a Markovian evolution.

In order to transmit $M$ classical messages perfectly through a channel $\Phi:\B{\Hil}\to \B{\Hil}$,  one must encode the $M$ messages in quantum states $\{\rho_m\}_{m=1}^M\subset \B{\cH}$ such that 
\begin{equation}\label{eq:Cencoding}
     \forall m\neq m': \,\, \Phi(\rho_m) \perp \Phi(\rho_{m'}),
\end{equation}
where two states $\rho,\sigma$ are orthogonal ($\rho\perp \sigma$) if their supports are orthogonal as subspaces. The interpretation here is that for any choice of encoding of the $M$ classical messages on the input side, the set of output states have to be perfectly distinguishable in order for the receiver to decode the intended message via a measurement without error, which is possible if and only if the output states are pairwise orthogonal. The maximum number of bits that can be transmitted in this fashion through $\Phi$ is called the \emph{(one-shot) zero-error classical capacity} of $\Phi$. Similarly, in order to send an $M-$dimensional quantum system perfectly through $\Phi$, one must find an encoding subspace $\mathcal{C}\subseteq \Hil$ with $\dim \mathcal{C}=M$ such that there exists a recovery channel $\mathcal{R}$ satisfying
\begin{equation}\label{eq:Qencoding}
    \mathcal{R}\circ \Phi (\rho) = \rho
\end{equation}
for all quantum states $\rho$ supported inside $\mathcal{C}$. The maximum number of qubits that can be transmitted in this fashion through $\Phi$ is called the \emph{(one-shot) zero-error quantum capacity} of $\Phi$. We denote these capacities in the classical and quantum case, respectively, by
\begin{equation}
    C^{(1)}_0(\Phi) \quad \text{and} \quad Q^{(1)}_0(\Phi).
\end{equation}

%denoted $C^{(1)}_0(\Phi)$)(denoted $Q^{(1)}_0(\Phi)$). 

\subsection{Main results}
We now summarize the primary contribution of our work. Consider a quantum system $A$ whose time evolution is governed by a dQMS $\{\Phi^n \}_{n\in \mathbb N}$, where $\Phi:\cL (\Hil)\to \B{\Hil}$ is a quantum channel and $d=\dim \Hil$. We address the following questions/problems in this paper:

\begin{itemize}
    \item Does there exist a finite time $n \in {\mathbb{N}}$ at which the one-shot zero-error classical (resp.~quantum) capacity of $\Phi^n$ vanishes? If yes, we say that the dQMS $\{\Phi^n \}_{n\in \mathbb N}$ is \emph{eventually c-scrambling (resp. q-scrambling)}. In this case, the dynamics is so noisy that no matter how cleverly we encode information in $A$, eventually, we will not be able to perfectly recover it. 
    \item For an eventually scrambling dQMS $\{\Phi^n \}_{n\in \mathbb N}$, we denote the  minimum time $n\in \mathbb N$ at which $\Phi^n$ loses its ability to perfectly transmit classical (resp. quantum) information by $c(\Phi)$ (resp. $q(\Phi)$) and refer to it as the \emph{classical (resp.~quantum) scrambling time}
    (or the \emph{scrambling index}) of $\Phi$. This is the minimum time after which any encoded information in the system will get `scrambled' beyond the point of perfect recovery. Finding bounds on the scrambling times $c(\Phi)$ and $ q(\Phi)$ of the dynamics is a natural problem to consider.
    \item Finally, if the dQMS is not eventually scrambling, what is the optimal way to encode information in such a way that it is protected from noise for an arbitrarily long time? 
\end{itemize}

Our main results provide full solutions to all of the above. Firstly, we completely characterize the class of eventually scrambling dQMS. 

\begin{theorem}\label{thm:main1}
    A dQMS $\{\Phi^n \}_{n\in \mathbb N}$ governing the dynamics of an open quantum system is
    \begin{itemize}
         \item eventually q-scrambling if and only if it is \emph{asymptotically entanglement breaking}, i.e., if and only if all the limit points of the semigroup $\{\Phi^n \}_{n\in \mathbb{N}}$ are entanglement breaking. 
        \item eventually c-scrambling if and only if it is \emph{mixing}, i.e., if and only if there exists a state $\rho\in \B{\Hil}$ such that 
    \begin{equation}
        \forall X \in \B{\Hil }: \quad \lim_{n\to \infty} \Phi^n (X) = \Tr (X) \rho.
    \end{equation}
    \end{itemize}  
    
\end{theorem}

Note that Theorem~\ref{thm:main1} provides an information theoretic interpretation to the entanglement-breaking and mixing behaviours of dQMS, which have been extensively studied in the literature \cite{Burgarth2013ergodic, Lami2016eb, eric2020eb}.

Secondly, we provide a universal upper bound on the scrambling times of all eventually scrambling dQMS that scales quadratically with the dimension of the system. Moreover, we show that this quadratic dependence is optimal in the classical case by exhibiting an explicit class of dQMS whose classical scrambling time scales quadratically with the system dimension. 

\begin{theorem}\label{thm:main2}
    For an eventually c-scrambling dQMS $\{ \Phi^n\}_{n\in \mathbb N}$ governing the dynamics of a $d-$dimensional system, the scrambling times satisfy $q(\Phi)\leq c(\Phi)\leq d^2$. Moreover, there exists a dQMS $\{ \Phi^n\}_{n\in \mathbb N}$ acting on a $d-$dimensional system such that
    \begin{equation*}
    c(\Phi) = \ceil[\bigg]{\frac{d^2 -2d +2}{2}}.
    \end{equation*}
\end{theorem}

\begin{theorem}\label{thm:main22}
    For an eventually q-scrambling dQMS $\{ \Phi^n\}_{n\in \mathbb N}$ governing the dynamics of a $d-$dimensional system, the scrambling time satisfies $q(\Phi)\leq d^2$. 
\end{theorem}

The above results are special cases of a more general phenomenon. It turns out that for a $d-$dimensional quantum system, there exists a universal dimension-dependent time-scale after which the zero-error capacities of any dQMS acting on the system stabilize.

\begin{theorem}\label{thm:main3}
    For any dQMS $\{ \Phi^n\}_{n\in \mathbb N}$ governing the dynamics of a $d-$dimensional quantum system, there exists $N\leq d^2$ such that
    \begin{align}
        \forall n\in \mathbb{N}: \quad 
         Q^{(1)}_0(\Phi^N) &= Q^{(1)}_0(\Phi^{N+n}), \\ C^{(1)}_0(\Phi^N) &= C^{(1)}_0(\Phi^{N+n}).
    \end{align}
\end{theorem}

Finally, if the dynamics $\{\Phi^n \}_{n\in \mathbb N}$ of the system is not eventually scrambling, we show that the optimal way to encode information in order to protect it from noise is to do it inside the peripheral space $\chi (\Phi)$ of the channel $\Phi$, which is defined as the span of all its peripheral eigenoperators:
\begin{equation}
    \chi(\Phi):= \operatorname{span} \{X\in \B{\Hil} : \Phi (X) = \lambda X, \abs{\lambda}=1 \}.
\end{equation}

The structure of this space is well understood: for any channel $\Phi:\B{\Hil}\to \B{\Hil}$, there exists an orthogonal decomposition $\Hil = \Hil_0 \oplus \bigoplus_{k=1}^K \Hil_{k,1}\otimes \Hil_{k,2}$ and positive definite states $\rho_{k}\in \B{\Hil_{k,2}}$ such that \cite[Theorem 6.16]{Wolf2012Qtour}:
    \begin{equation}\label{eq:phasespace*}
        \chi(\Phi) = 0 \oplus \bigoplus_{k=1}^K (\B{\Hil_{k,1}}\otimes \rho_k). 
    \end{equation}
Moreover, there exist unitaries $U_k\in \B{\Hil_{k,1}}$ and a permutation $\pi$ which permutes within subsets of $\{1,2,\ldots ,K \}$ for which the corresponding $\Hil_{k,1}$'s have the same dimension, such that for any
    \begin{equation}
        X = 0 \oplus \bigoplus_{k=1}^K x_k \otimes \rho_k, \quad
        \Phi (X) = 0 \oplus \bigoplus_{k=1}^K U^{\dagger}_k x_{\pi (k)} U_k \otimes \rho_k .
    \end{equation}

From the above peripheral decomposition, it is not too hard to deduce that any information encoded inside the $\B{\Hil_{k,1}}$ blocks is shielded from noise for an arbitrarily long time. Furthermore, in the asymptotic limit, it turns out that this is the best one can do. 

\begin{theorem}\label{thm:main4}
    For any dQMS $\{\Phi^n \}_{n\in \mathbb N}$, we have
    \begin{align}
        \lim_{n\to \infty} C^{(1)}_0 (\Phi^n) &= 
        \log \sum_{k=1}^K \dim \Hil_{k,1}. \\
        \lim_{n\to \infty} Q^{(1)}_0 (\Phi^n) &= 
        \log \max_{k} \dim \Hil_{k,1}.
    \end{align}
%    Moreover, the bound gets saturated for a large enough $n$.
\end{theorem}

Note that because of Theorem~\ref{thm:main3}, the limits above are actually attained at a finite time $n\leq d^2$. Two special cases of this result are worth highlighting. If $\Phi$ is a classical channel given by a stochastic matrix $M$:
\begin{equation}
    \Phi(X) = \sum_{i,j} M_{ij} X_{jj} \ketbra{i},
\end{equation}
the $\Hil_{k,1}$ blocks in Eq.~\eqref{eq:phasespace*} become one-dimensional (since otherwise the channel would have non-zero quantum capacity, which is impossible). Hence, $\sum_k \dim \Hil_{k,1} = \dim \chi (\Phi) =$ number of peripheral eigenvalues of $M$ (counted with multiplicities). Similarly, if $\Phi$ is a quantum channel with a unique fixed state, the $\Hil_{k,1}$ blocks become one-dimensional \cite{Wolf2012Qtour}, and we get

\begin{equation*}
        \lim_{n\to \infty} C^{(1)}_0 (\Phi^n) = \log \dim \chi (\Phi).
\end{equation*}

\begin{remark}
    The long-time capacities of a continuous-time Quantum Markov Semigroup $\{\Psi_t = e^{t\mathcal{L}} \}_{t\geq 0}$ generated by a Lindbladian $\mathcal{L}$ \cite{Gorini1976qms, Lindblad1976qms} can be calculated by choosing $\Phi=\Psi_1=e^{\mathcal{L}}$, for instance, and applying the formulas from Theorem~\ref{thm:main4} for the dQMS $\{\Phi^n \}_{n\in \mathbb{N}}$. It is easy to see that the choice $t=1$ is irrelevant here, since the peripheral space $\chi(\Psi_t)$ is independent of $t$ \cite{fawzi2024error}.
\end{remark}

%\begin{theorem}\label{thm:main5}
%    Let $\{ \Phi^n\}_{n\in \mathbb N}$ be a dQMS such that $\Phi$ has a unique fixed state. Then,
%\end{theorem} 

\subsection{Proof ideas}
Two key ingredients are used in the proofs of our results. The first is a reformulation of the zero-error capacity of any channel $\Phi$ in terms of its operator system. If $\{K_i \}_i$ is a set of Kraus operators of $\Phi:\B{\Hil}\to \B{\Hil}$, the operator system $S_{\Phi}:= \text{span}_{i,j} \{K_i^{\dagger}K_j \}$ is a $\dagger-$closed subspace of $\B{\Hil}$ containing the identity \cite{paulsen-book}. The error-correction condition in Eq.~\eqref{eq:Cencoding} can be restated as 
\begin{align}
    \forall m\neq m' : \quad \ketbra{\psi_m}{\psi_{m'}} \perp S_{\Phi} \quad \text{\cite{Duan2013noncomm}}.
\end{align}
A similar reformulation can be done for the quantum case in Eq.~\eqref{eq:Qencoding} by exploiting the Knill-Laflamme error correction conditions \cite{knil-laf}. Hence, the zero-error capacities of any channel $\Phi$ are purely a function of its operator system. Now, for any dQMS $\{ \Phi^n\}_{n\in \mathbb{N}}$, we show that the corresponding operator systems form an increasing chain of subspaces in $\B{\Hil}$ which stabilizes at time $N\leq (\dim \Hil )^2$:
\begin{equation}
    S_{\Phi} \subset S_{\Phi^2} \subset \ldots \subset S_{\Phi^N} = S_{\Phi^{N+1}} = \ldots = S_{\Phi^{N+n}} = \ldots .
\end{equation}
This proves that the zero-error capacities must also stabilize after time $N$ (Theorem~\ref{thm:main3}). The second key ingredient we use is the fact that for any dQMS $\{ \Phi^n\}_{n\in \mathbb{N}}$, there exists an increasing subsequence $(n_i)_{i\in \mathbb N}$ such that $\Phi^{n_i}\to \mathcal{P}_{\chi}$ as $i\to \infty$ \cite[Proposition 6.3]{Wolf2012Qtour}, where $P_{\chi}$ is the channel that projects onto the peripheral space $\chi (\Phi)$. This shows that the capacities of $\Phi^n$ must stabilize to the corresponding capacities of $P_{\chi}$, which can be explicitly computed in terms of the block structure of $\chi (\Phi)$ (see Eq.~\eqref{eq:phasespace*}). The characterization of eventually scrambling dQMS follows easily from this. Clearly, a dQMS $\{\Phi^n \}_{n\in \mathbb N}$ is eventually c-scrambling $\iff \sum_{k=1}^K \dim \Hil_{k,1}=1 \iff \dim \chi (\Phi)=1\iff \Phi$ admits a unique fixed state and no other peripheral eigenoperators, which is equivalent to the dQMS being mixing in the sense of Theorem~\ref{thm:main1}. The quantum case follows similarly by exploiting the results derived in \cite{Lami2016eb} on asymptotically entanglement-breaking channels. We refer the reader to Appendix~\ref{sec:main} for complete proofs of all the main results.

\subsection{Auxiliary Results}\label{sec:main-res}

We derive several auxiliary results in the appendices, which might be of independent interest. We list a few of them below. 

\begin{itemize}
    \item We prove that if a dQMS $\{\Phi^n \}_{n\in \mathbb N}$ is eventually c-scrambling, then the one-shot zero error \emph{entanglement assisted} classical capacity of $\Phi^n$ also vanishes for some $n\in \mathbb N$, see Theorem~\ref{theorem:scr-mix}. Moreover, the entanglement assisted classical scrambling time $c_E(\Phi)$, defined analogously to $c(\Phi)$, also satisfies $c_E(\Phi)\leq d^2$, see Theorem~\ref{stabilize} and Corollary~\ref{corollary:index}.
    \item We provide an equivalent characterization of mixing channels in terms of minimal invariant subspaces\footnote{A minimal invariant subspace is a subspace of the
    Hilbert space that is left invariant by the action of the channel and does not contain
    any non-trivial, smaller subspace with the same property}, see Theorem~\ref{thm:inv-sub}. This result is a generalization of~\cite[Proposition 3]{wielandt2010} where such a characterization was obtained for the smaller class of \emph{primitive}\footnote{A mixing channel is called primitive if its unique fixed state has full rank.} channels.
    \item We exhibit a close link between the Wielandt index \cite{wielandt2010, rahaman2020, Wiel-revisit, jia-capel} of a channel and its scrambling time. It turns out that if $\Phi$ is a primitive channel, there exists $n\in \mathbb N$ such that $\Phi^n$ is \emph{strictly positive} (i.e., it sends any input state to a full rank output state), and the minimum such $n$ is called the Wielandt index of $\Phi$ (denoted $w(\Phi)$). For any primitive channel, it is easy to check that $c(\Phi)\leq w(\Phi)$. In case $\Phi$ is also unital, we obtain a reverse inequality: $w(\Phi)\leq (d-1)c(\Phi)$, see Corollary~\ref{strict-pos}.
    \item We show that for the so-called diagonal unitary covariant channels~\cite{Singh2021diagonalunitary}, the properties of strict positivity, scrambling, primitivity, and mixing are equivalent to the corresponding properties of a classical stochastic matrix, see Theorems~\ref{thm:strict-pos-DUC}, \ref{thm:scrambling-DUC}, \ref{thm:DUC-primitive}, and \ref{thm:DUC-mixing}. 
    
\end{itemize}

\section{Discussion and Outlook}

In this work, we analyze how information can be transmitted across sequential concatenations of a quantum channel $\Phi$. Physically speaking, we investigate how information stored in an open quantum system propagates over time as the system evolves according to a discrete-time Quantum Markov Semigroup $\{\Phi^n \}_{n\in \mathbb N}$. We show that any information stored inside the peripheral space $\chi(\Phi)$ is protected from noise for all times $n\in \mathbb{N}$. Furthermore, we prove that this is the optimal encoding strategy in the limit of $n\to \infty$. This allows us to derive explicit formulas for the information transmission capacities of any dQMS in the limit of $n\to \infty$ in terms of the block structure of $\chi (\Phi)$. Here, it would be interesting to analyze whether the decomposition of $\chi(\Phi)$ in Eq.~\eqref{eq:phasespace*} is efficiently computable, which in turn would make the capacity formulas of Theorem~\ref{thm:main4} efficiently computable. We also show that a system is asymptotically useless for storing classical (resp. quantum) information if and only if the dQMS governing its dynamics is mixing (resp. asymptotically entanglement breaking). Interestingly, we exhibit a universal time scale $n\leq d^2$, after which the information transmission capacity of any dQMS acting on a $d-$dimensional system stabilizes. We prove that the quadratic dimension dependence is tight for classical capacity and we expect the same to be true for quantum capacity. 

The sequential view of information transmission that we consider opens a host of exciting research directions. While we have considered a simple Markovian model for the dynamics of the open system, it would be interesting to perform the same analysis for other kinds of dynamics, such as repeated interaction systems \cite{Attal2006repeated, Grimmer2016repeated, Ciccarello2022repeated} and other non-Markovian dynamics. Apart from the standard classical and quantum capacities, other kinds of transmission rates can also be considered, such as those where assistance from external resources like correlations and entanglement are supplied to aid in communication \cite{Cubitt2010zero, Cubitt2011zeroextra, Leung2012zeroextra, Duan2016zeroextra}. Finally, it would be interesting to drop the zero-error constraint and analyse approximate capacities, where information is required to be recovered only approximately with a certain error threshold.

\textbf{Note.} Many of the open questions posed in this discussion have been recently resolved, see \cite{fawzi2024error, singh2024markovian, singh2024markovian2}.

\section{Acknowledgements}
We thank Omar Fawzi for insightful comments on the first version of this manuscript and for conjecturing that the results in Theorem~\ref{thm:main4} should hold true.

At the time of this work, SS was supported by the Cambridge Trust International Scholarship. At the time of this work, MR was supported by the Marie Sk\l{}odowska-Curie Fellowship from the European Union’s Horizon Research and Innovation program, grant Agreement No. HORIZON-MSCA-2022-PF-01 (Project number: 101108117).

\smallskip 

\bibliographystyle{plainurl}
\bibliography{references}

\begin{thebibliography}{10}

\bibitem{Akelbek2009index}
Mahmud Akelbek and Steve Kirkland.
\newblock Coefficients of ergodicity and the scrambling index.
\newblock {\em Linear Algebra and its Applications}, 430(4):1111–1130, February 2009.
\newblock URL: \url{http://dx.doi.org/10.1016/j.laa.2008.10.007}, \href {https://doi.org/10.1016/j.laa.2008.10.007} {\path{doi:10.1016/j.laa.2008.10.007}}.

\bibitem{Alicki2002QDS}
R.~Alicki.
\newblock {\em Invitation to Quantum Dynamical Semigroups}, page 239–264.
\newblock Springer Berlin Heidelberg, 2002.
\newblock URL: \url{http://dx.doi.org/10.1007/3-540-46122-1_10}, \href {https://doi.org/10.1007/3-540-46122-1_10} {\path{doi:10.1007/3-540-46122-1_10}}.

\bibitem{Alicki2007QDS}
Robert Alicki and Karl Lendi.
\newblock {\em Quantum Dynamical Semigroups and Applications}.
\newblock Springer Berlin Heidelberg, 2007.
\newblock URL: \url{http://dx.doi.org/10.1007/3-540-70861-8}, \href {https://doi.org/10.1007/3-540-70861-8} {\path{doi:10.1007/3-540-70861-8}}.

\bibitem{Attal2006repeated}
Stéphane Attal and Yan Pautrat.
\newblock From repeated to continuous quantum interactions.
\newblock {\em Annales Henri Poincaré}, 7(1):59–104, January 2006.
\newblock URL: \url{http://dx.doi.org/10.1007/s00023-005-0242-8}, \href {https://doi.org/10.1007/s00023-005-0242-8} {\path{doi:10.1007/s00023-005-0242-8}}.

\bibitem{Breuer2007open}
Heinz-Peter Breuer and Francesco Petruccione.
\newblock {\em The Theory of Open Quantum Systems}.
\newblock Oxford University PressOxford, January 2007.
\newblock URL: \url{http://dx.doi.org/10.1093/acprof:oso/9780199213900.001.0001}, \href {https://doi.org/10.1093/acprof:oso/9780199213900.001.0001} {\path{doi:10.1093/acprof:oso/9780199213900.001.0001}}.

\bibitem{Burgarth2013ergodic}
D~Burgarth, G~Chiribella, V~Giovannetti, P~Perinotti, and K~Yuasa.
\newblock Ergodic and mixing quantum channels in finite dimensions.
\newblock {\em New Journal of Physics}, 15(7), July 2013.
\newblock URL: \url{http://dx.doi.org/10.1088/1367-2630/15/7/073045}, \href {https://doi.org/10.1088/1367-2630/15/7/073045} {\path{doi:10.1088/1367-2630/15/7/073045}}.

\bibitem{Choi2009error}
Man-Duen Choi, Nathaniel Johnston, and David~W Kribs.
\newblock The multiplicative domain in quantum error correction.
\newblock {\em Journal of Physics A: Mathematical and Theoretical}, 42(24):245303, May 2009.
\newblock URL: \url{http://dx.doi.org/10.1088/1751-8113/42/24/245303}, \href {https://doi.org/10.1088/1751-8113/42/24/245303} {\path{doi:10.1088/1751-8113/42/24/245303}}.

\bibitem{Ciccarello2022repeated}
Francesco Ciccarello, Salvatore Lorenzo, Vittorio Giovannetti, and G.~Massimo Palma.
\newblock Quantum collision models: Open system dynamics from repeated interactions.
\newblock {\em Physics Reports}, 954:1–70, April 2022.
\newblock URL: \url{http://dx.doi.org/10.1016/j.physrep.2022.01.001}, \href {https://doi.org/10.1016/j.physrep.2022.01.001} {\path{doi:10.1016/j.physrep.2022.01.001}}.

\bibitem{Cubitt2010zero}
Toby~S. Cubitt, Debbie Leung, William Matthews, and Andreas Winter.
\newblock Improving zero-error classical communication with entanglement.
\newblock {\em Phys. Rev. Lett.}, 104:230503, Jun 2010.
\newblock URL: \url{https://link.aps.org/doi/10.1103/PhysRevLett.104.230503}, \href {https://doi.org/10.1103/PhysRevLett.104.230503} {\path{doi:10.1103/PhysRevLett.104.230503}}.

\bibitem{Cubitt2011zeroextra}
Toby~S. Cubitt, Debbie Leung, William Matthews, and Andreas Winter.
\newblock Zero-error channel capacity and simulation assisted by non-local correlations.
\newblock {\em IEEE Transactions on Information Theory}, 57(8):5509–5523, August 2011.
\newblock URL: \url{http://dx.doi.org/10.1109/TIT.2011.2159047}, \href {https://doi.org/10.1109/tit.2011.2159047} {\path{doi:10.1109/tit.2011.2159047}}.

\bibitem{davidson1996calgebra}
K.R. Davidson.
\newblock {\em C*-Algebras by Example}.
\newblock Fields Institute for Research in Mathematical Sciences Toronto: Fields Institute monographs. American Mathematical Society, 1996.
\newblock URL: \url{https://books.google.co.uk/books?id=PjpgCgAAQBAJ}.

\bibitem{Duan2009zero}
Runyao Duan.
\newblock Super-activation of zero-error capacity of noisy quantum channels.
\newblock {\em preprint arXiv:0906.2527}, 2009.
\newblock URL: \url{https://arxiv.org/abs/0906.2527}.

\bibitem{Duan2013noncomm}
Runyao Duan, Simone Severini, and Andreas Winter.
\newblock Zero-error communication via quantum channels, noncommutative graphs, and a quantum \text{Lovász Number}.
\newblock {\em IEEE Transactions on Information Theory}, 59(2):1164--1174, 2013.
\newblock \href {https://doi.org/10.1109/TIT.2012.2221677} {\path{doi:10.1109/TIT.2012.2221677}}.

\bibitem{Duan2016zeroextra}
Runyao Duan and Andreas Winter.
\newblock No-signalling-assisted zero-error capacity of quantum channels and an information theoretic interpretation of the lovász number.
\newblock {\em IEEE Transactions on Information Theory}, 62(2):891–914, February 2016.
\newblock URL: \url{http://dx.doi.org/10.1109/TIT.2015.2507979}, \href {https://doi.org/10.1109/tit.2015.2507979} {\path{doi:10.1109/tit.2015.2507979}}.

\bibitem{Evans1978Spectral}
David~E. Evans and Raphael Høegh-Krohn.
\newblock Spectral properties of positive maps on c* -algebras.
\newblock {\em Journal of the London Mathematical Society}, s2-17(2):345–355, April 1978.
\newblock URL: \url{http://dx.doi.org/10.1112/jlms/s2-17.2.345}, \href {https://doi.org/10.1112/jlms/s2-17.2.345} {\path{doi:10.1112/jlms/s2-17.2.345}}.

\bibitem{Fagnola1999semigroup}
Franco Fagnola.
\newblock Quantum markov semigroups.
\newblock {\em Proyecciones (Antofagasta)}, 18(3):29–74, 1999.
\newblock URL: \url{http://dx.doi.org/10.22199/S07160917.1999.0003.00004}, \href {https://doi.org/10.22199/s07160917.1999.0003.00004} {\path{doi:10.22199/s07160917.1999.0003.00004}}.

\bibitem{fawzi2024error}
Omar Fawzi, Mizanur Rahaman, and Mostafa Taheri.
\newblock Capacities of quantum markovian noise for large times.
\newblock {\em preprint arXiv:2408.00116}, 2024.
\newblock URL: \url{https://arxiv.org/abs/2408.00116}.

\bibitem{Gorini1976qms}
Vittorio Gorini, Andrzej Kossakowski, and E.~C.~G. Sudarshan.
\newblock Completely positive dynamical semigroups of n-level systems.
\newblock {\em Journal of Mathematical Physics}, 17(5):821–825, May 1976.
\newblock URL: \url{http://dx.doi.org/10.1063/1.522979}, \href {https://doi.org/10.1063/1.522979} {\path{doi:10.1063/1.522979}}.

\bibitem{Grimmer2016repeated}
Daniel Grimmer, David Layden, Robert~B. Mann, and Eduardo Martín-Martínez.
\newblock Open dynamics under rapid repeated interaction.
\newblock {\em Physical Review A}, 94(3), September 2016.
\newblock URL: \url{http://dx.doi.org/10.1103/PhysRevA.94.032126}, \href {https://doi.org/10.1103/physreva.94.032126} {\path{doi:10.1103/physreva.94.032126}}.

\bibitem{Guterman2019index}
A.~E. Guterman and A.~M. Maksaev.
\newblock Upper bounds on scrambling index for non-primitive digraphs.
\newblock {\em Linear and Multilinear Algebra}, 69(11):2143–2168, September 2019.
\newblock URL: \url{http://dx.doi.org/10.1080/03081087.2019.1663139}, \href {https://doi.org/10.1080/03081087.2019.1663139} {\path{doi:10.1080/03081087.2019.1663139}}.

\bibitem{eric2020eb}
Eric~P. Hanson, Cambyse Rouzé, and Daniel Stilck~Fran\c{c}a.
\newblock Eventually entanglement breaking markovian dynamics: Structure and characteristic times.
\newblock {\em Annales Henri Poincaré}, 21(5):1517–1571, March 2020.
\newblock URL: \url{http://dx.doi.org/10.1007/s00023-020-00906-4}, \href {https://doi.org/10.1007/s00023-020-00906-4} {\path{doi:10.1007/s00023-020-00906-4}}.

\bibitem{Hiai2011sufficient}
Fumio Hiai, Mil\'an Mosonyi, D\'enez Petz, and C\'edric B\'eny.
\newblock Quantum f-divergences and error correction.
\newblock {\em Reviews in Mathematical Physics}, 23(07):691–747, August 2011.
\newblock URL: \url{http://dx.doi.org/10.1142/S0129055X11004412}, \href {https://doi.org/10.1142/s0129055x11004412} {\path{doi:10.1142/s0129055x11004412}}.

\bibitem{Hiai2015contraction}
Fumio Hiai and Mary~Beth Ruskai.
\newblock Contraction coefficients for noisy quantum channels.
\newblock {\em Journal of Mathematical Physics}, 57(1), December 2015.
\newblock URL: \url{http://dx.doi.org/10.1063/1.4936215}, \href {https://doi.org/10.1063/1.4936215} {\path{doi:10.1063/1.4936215}}.

\bibitem{idel}
Martin Idel.
\newblock On the structure of positive maps.
\newblock {\em Master’s thesis, Technische Universitat Munchen}, 2013.

\bibitem{jaques-rahaman2017}
Samuel Jaques and Mizanur Rahaman.
\newblock Spectral properties of tensor products of channels.
\newblock {\em Journal of Mathematical Analysis and Applications}, 465(2):1134–1158, September 2018.
\newblock URL: \url{http://dx.doi.org/10.1016/j.jmaa.2018.05.052}, \href {https://doi.org/10.1016/j.jmaa.2018.05.052} {\path{doi:10.1016/j.jmaa.2018.05.052}}.

\bibitem{Jenov2006sufficent}
Anna Jenčová and Dénes Petz.
\newblock Sufficiency in quantum statistical inference.
\newblock {\em Communications in Mathematical Physics}, 263(1):259–276, January 2006.
\newblock URL: \url{http://dx.doi.org/10.1007/s00220-005-1510-7}, \href {https://doi.org/10.1007/s00220-005-1510-7} {\path{doi:10.1007/s00220-005-1510-7}}.

\bibitem{Jenov2006sufficient2}
Anna Jenčová and Dénes Petz.
\newblock Sufficiency in quantum statistical inference: A survey with examples.
\newblock {\em Infinite Dimensional Analysis, Quantum Probability and Related Topics}, 09(03):331–351, September 2006.
\newblock URL: \url{http://dx.doi.org/10.1142/S0219025706002408}, \href {https://doi.org/10.1142/s0219025706002408} {\path{doi:10.1142/s0219025706002408}}.

\bibitem{jia-capel}
Yifan Jia and Angela Capel.
\newblock A generic quantum {W}ielandt inequality.
\newblock {\em Quantum}, 8:1331, May 2024.
\newblock URL: \url{http://dx.doi.org/10.22331/q-2024-05-02-1331}, \href {https://doi.org/10.22331/q-2024-05-02-1331} {\path{doi:10.22331/q-2024-05-02-1331}}.

\bibitem{khatri2024principles}
Sumeet Khatri and Mark~M. Wilde.
\newblock Principles of quantum communication theory: A modern approach, 2024.
\newblock URL: \url{https://arxiv.org/abs/2011.04672}, \href {https://arxiv.org/abs/2011.04672} {\path{arXiv:2011.04672}}.

\bibitem{knil-laf}
Emanuel Knill and Raymond Laflamme.
\newblock Theory of quantum error-correcting codes.
\newblock {\em Phys. Rev. A}, 55:900--911, Feb 1997.
\newblock URL: \url{https://link.aps.org/doi/10.1103/PhysRevA.55.900}, \href {https://doi.org/10.1103/PhysRevA.55.900} {\path{doi:10.1103/PhysRevA.55.900}}.

\bibitem{KKS}
Dennis Kretschmann, David~W. Kribs, and Robert~W. Spekkens.
\newblock Complementarity of private and correctable subsystems in quantum cryptography and error correction.
\newblock {\em Physical Review A}, 78(3), September 2008.
\newblock URL: \url{http://dx.doi.org/10.1103/PhysRevA.78.032330}, \href {https://doi.org/10.1103/physreva.78.032330} {\path{doi:10.1103/physreva.78.032330}}.

\bibitem{KSW}
Dennis Kretschmann, Dirk Schlingemann, and Reinhard~F. Werner.
\newblock The information-disturbance tradeoff and the continuity of \text{Stinespring's} representation.
\newblock {\em IEEE Transactions on Information Theory}, 54(4):1708--1717, 2008.
\newblock \href {https://doi.org/10.1109/TIT.2008.917696} {\path{doi:10.1109/TIT.2008.917696}}.

\bibitem{kribs}
David~W. {Kribs}.
\newblock {Quantum Channels, Wavelets, Dilations and Representations of $O_n$}.
\newblock {\em Proc. Edinburgh Math. Soc. 46 (2003), 421-433}, sep 2003.
\newblock \href {https://arxiv.org/abs/math/0309390} {\path{arXiv:math/0309390}}, \href {https://doi.org/doi:10.1017/S0013091501000980} {\path{doi:doi:10.1017/S0013091501000980}}.

\bibitem{kribs2006error}
David~W. Kribs, Raymond Laflamme, David Poulin, and Maia Lesosky.
\newblock Operator quantum error correction.
\newblock {\em Quantum Inf. Comput.}, 6(4):382--399, 2006.
\newblock \href {https://doi.org/10.26421/QIC6.4-5-6} {\path{doi:10.26421/QIC6.4-5-6}}.

\bibitem{Lami2016eb}
L.~Lami and V.~Giovannetti.
\newblock Entanglement-saving channels.
\newblock {\em Journal of Mathematical Physics}, 57(3), March 2016.
\newblock URL: \url{http://dx.doi.org/10.1063/1.4942495}, \href {https://doi.org/10.1063/1.4942495} {\path{doi:10.1063/1.4942495}}.

\bibitem{Leung2012zeroextra}
Debbie Leung, Laura Mancinska, William Matthews, Maris Ozols, and Aidan Roy.
\newblock Entanglement can increase asymptotic rates of zero-error classical communication over classical channels.
\newblock {\em Communications in Mathematical Physics}, 311(1):97–111, March 2012.
\newblock URL: \url{http://dx.doi.org/10.1007/s00220-012-1451-x}, \href {https://doi.org/10.1007/s00220-012-1451-x} {\path{doi:10.1007/s00220-012-1451-x}}.

\bibitem{Lindblad1976qms}
G.~Lindblad.
\newblock On the generators of quantum dynamical semigroups.
\newblock {\em Communications in Mathematical Physics}, 48(2):119–130, June 1976.
\newblock URL: \url{http://dx.doi.org/10.1007/BF01608499}, \href {https://doi.org/10.1007/bf01608499} {\path{doi:10.1007/bf01608499}}.

\bibitem{Wiel-revisit}
Mateusz Michałek and Yaroslav Shitov.
\newblock Quantum version of {W}ielandt’s inequality revisited.
\newblock {\em IEEE Transactions on Information Theory}, 65(8):5239--5242, 2019.
\newblock \href {https://doi.org/10.1109/TIT.2019.2897772} {\path{doi:10.1109/TIT.2019.2897772}}.

\bibitem{paulsen-book}
Vern Paulsen.
\newblock {\em Completely Bounded Maps and Operator Algebras}.
\newblock Cambridge University Press, February 2003.
\newblock URL: \url{http://dx.doi.org/10.1017/CBO9780511546631}, \href {https://doi.org/10.1017/cbo9780511546631} {\path{doi:10.1017/cbo9780511546631}}.

\bibitem{Petz1988sufficient}
D\'enes Petz.
\newblock Sufficiency of channels over von neumann algebras.
\newblock {\em The Quarterly Journal of Mathematics}, 39(1):97–108, 1988.
\newblock URL: \url{http://dx.doi.org/10.1093/qmath/39.1.97}, \href {https://doi.org/10.1093/qmath/39.1.97} {\path{doi:10.1093/qmath/39.1.97}}.

\bibitem{Preskill2018nisq}
John Preskill.
\newblock Quantum computing in the nisq era and beyond.
\newblock {\em Quantum}, 2:79, August 2018.
\newblock URL: \url{http://dx.doi.org/10.22331/q-2018-08-06-79}, \href {https://doi.org/10.22331/q-2018-08-06-79} {\path{doi:10.22331/q-2018-08-06-79}}.

\bibitem{rahaman2017}
Mizanur {Rahaman}.
\newblock {Multiplicative properties of quantum channels}.
\newblock {\em Journal of Physics A Mathematical General}, 50(34):345302, August 2017.
\newblock \href {https://arxiv.org/abs/1701.06205} {\path{arXiv:1701.06205}}, \href {https://doi.org/10.1088/1751-8121/aa7b57} {\path{doi:10.1088/1751-8121/aa7b57}}.

\bibitem{rahaman2020}
Mizanur Rahaman.
\newblock A new bound on quantum {W}ielandt inequality.
\newblock {\em IEEE Transactions on Information Theory}, 66(1):147--154, 2020.
\newblock \href {https://doi.org/10.1109/TIT.2019.2945776} {\path{doi:10.1109/TIT.2019.2945776}}.

\bibitem{Ruskai1994contraction}
Mary~Beth Ruskai.
\newblock Beyond strong subadditivity? improved bounds on the contraction of generalized relative entropy.
\newblock {\em Reviews in Mathematical Physics}, 06(05a):1147–1161, January 1994.
\newblock URL: \url{http://dx.doi.org/10.1142/S0129055X94000407}, \href {https://doi.org/10.1142/s0129055x94000407} {\path{doi:10.1142/s0129055x94000407}}.

\bibitem{wielandt2010}
Mikel Sanz, David Pérez-García, Michael~M. Wolf, and Juan~I. Cirac.
\newblock A quantum version of \text{Wielandt's} inequality.
\newblock {\em IEEE Transactions on Information Theory}, 56(9):4668--4673, 2010.
\newblock \href {https://doi.org/10.1109/TIT.2010.2054552} {\path{doi:10.1109/TIT.2010.2054552}}.

\bibitem{Shannon1956zero}
C.~Shannon.
\newblock The zero error capacity of a noisy channel.
\newblock {\em IRE Transactions on Information Theory}, 2(3):8--19, 1956.
\newblock \href {https://doi.org/10.1109/TIT.1956.1056798} {\path{doi:10.1109/TIT.1956.1056798}}.

\bibitem{Shannon1948info}
C.~E. Shannon.
\newblock A mathematical theory of communication.
\newblock {\em Bell System Technical Journal}, 27(3):379–423, July 1948.
\newblock URL: \url{http://dx.doi.org/10.1002/j.1538-7305.1948.tb01338.x}, \href {https://doi.org/10.1002/j.1538-7305.1948.tb01338.x} {\path{doi:10.1002/j.1538-7305.1948.tb01338.x}}.

\bibitem{Shirokov-Shulman2103}
Maksim {Shirokov} and Tatiana {Shulman}.
\newblock {On superactivation of zero-error capacities and reversibility of a quantum channel}.
\newblock {\em Commun. Math. Phys. V.335}, pages 1159--1179, September 2015.
\newblock \href {https://doi.org/10.1007/s00220-015-2345-5} {\path{doi:10.1007/s00220-015-2345-5}}.

\bibitem{Shirokov2022}
Maksim~Evgenievich Shirokov.
\newblock Convergence criterion for the quantum relative entropy and its use.
\newblock {\em Matematicheskii Sbornik}, 213(12):137–174, 2022.
\newblock URL: \url{http://dx.doi.org/10.4213/sm9794}, \href {https://doi.org/10.4213/sm9794} {\path{doi:10.4213/sm9794}}.

\bibitem{singh2024markovian}
Satvik Singh and Nilanjana Datta.
\newblock Information transmission under markovian noise, 2024.
\newblock URL: \url{https://arxiv.org/abs/2409.17743}, \href {https://arxiv.org/abs/2409.17743} {\path{arXiv:2409.17743}}.

\bibitem{singh2024markovian2}
Satvik Singh and Nilanjana Datta.
\newblock Information storage and transmission under markovian noise, 2025.
\newblock URL: \url{https://arxiv.org/abs/2504.10436}, \href {https://arxiv.org/abs/2504.10436} {\path{arXiv:2504.10436}}.

\bibitem{Singh2022ergodic}
Satvik Singh, Nilanjana Datta, and Ion Nechita.
\newblock Ergodic theory of diagonal orthogonal covariant quantum channels.
\newblock {\em Letters in Mathematical Physics}, 114(5), October 2024.
\newblock URL: \url{http://dx.doi.org/10.1007/s11005-024-01864-2}, \href {https://doi.org/10.1007/s11005-024-01864-2} {\path{doi:10.1007/s11005-024-01864-2}}.

\bibitem{Singh2021diagonalunitary}
Satvik Singh and Ion Nechita.
\newblock Diagonal unitary and orthogonal symmetries in quantum theory.
\newblock {\em {Quantum}}, 5:519, August 2021.
\newblock \href {https://doi.org/10.22331/q-2021-08-09-519} {\path{doi:10.22331/q-2021-08-09-519}}.

\bibitem{watrous2018theory}
John Watrous.
\newblock {\em The Theory of Quantum Information}.
\newblock Cambridge University Press, April 2018.
\newblock URL: \url{http://dx.doi.org/10.1017/9781316848142}, \href {https://doi.org/10.1017/9781316848142} {\path{doi:10.1017/9781316848142}}.

\bibitem{Wielandt1950index}
H.~Wielandt.
\newblock {Unzerlegbare, nicht negative Matrizen}.
\newblock {\em Mathematische Zeitschrift}, 52:642--648, 1950.
\newblock \href {https://doi.org/10.1007/BF02230720} {\path{doi:10.1007/BF02230720}}.

\bibitem{Wilde2013quantum}
Mark~M. Wilde.
\newblock {\em Quantum Information Theory}.
\newblock Cambridge University Press, April 2013.
\newblock URL: \url{http://dx.doi.org/10.1017/CBO9781139525343}, \href {https://doi.org/10.1017/cbo9781139525343} {\path{doi:10.1017/cbo9781139525343}}.

\bibitem{Wolf2012Qtour}
M.~M. Wolf.
\newblock Quantum channels and operations: Guided tour.
\newblock {\em (unpublished)}, 2012.
\newblock URL: \url{https://mediatum.ub.tum.de/doc/1701036/1701036.pdf}.

\end{thebibliography}

\begin{appendices}

\section{Preliminaries}
In this paper, we always work with finite-dimensional Hilbert spaces and denote them by $\cH$. $\cL(\cH)$ denotes the algebra of linear operators acting on $\cH$. The set of quantum states (density matrices) on $\cH$ is denoted by $\cD(\cH) := \{ \rho \in \cL(\cH) \, : \, \rho \geq 0, \, \Tr(\rho) = 1\}$. The identity operator on $\cH$ is denoted by ${\iden} \in \cL(\cH)$. Note that if $\dim \cH = d$, $\cH \simeq \C{d}$ and $\cL(\cH) \simeq \M{d}$, where $\M{d}$ denotes the matrix algebra of all $d\times d$ complex matrices. We denote the Hilbert space associated to a quantum system $A$ by $\Hil_A$. Pure states of the system $A$ are denoted either by normalized kets $\ket{\psi} \in \Hil_A$ or by the corresponding rank one projections $\psi:=\ketbra{\psi} \in \cL(\Hil_A)$. 

A quantum channel $\Phi:\cL({\Hil_A})\to \cL(\Hil_B)$ is a linear, completely positive, and trace-preserving map. $\Phi$ is said to be unital if $\Phi(\iden_A) = {\iden}_B$. The adjoint of the channel $\Phi$, with respect to the Hilbert-Schmidt inner product, is the linear map $\Phi^*:\cL({\Hil_B})\to \cL(\Hil_A)$ defined through the relation $\Tr (\Phi^*(X) Y) = \Tr (X \Phi(Y))$, for any $X \in \cL(\Hil_A)$ and  $Y\in \cL(\Hil_B)$. It is completely positive and unital. In this paper, all logarithms are taken to base $2$. 

A discrete quantum Markov semigroup (dQMS) associated with a quantum channel $\Phi:\B{\Hil}\to \B{\Hil}$ is the sequence $\{\Phi^n\}_{n\in \mathbb N}$. Here, the `semigroup' terminology simply refers to the fact that the sequence $\{\Phi^n \}_{n\in \mathbb N}$ is closed with respect to compositions: $\Phi^n\circ \Phi^m = \Phi^{n+m}$. We can think of a dQMS as governing the time evolution of an open quantum system $A$, with $n\in \mathbb N$ acting as the discrete time parameter.  

\begin{rem}
    As mentioned in the Introduction, we are interested in studying zero-error communication through discrete quantum Markov semigroups. Hence, we will mostly focus on quantum channels $\Phi:\cL (\Hil_A)\to \cL (\Hil_B)$ whose input and output spaces are the same $\Hil_A = \Hil_B = \cH \simeq \C{d}$. Using the isomorphism $\cL (\C{d}) \simeq \M{d}$, we will interchangeably denote such channels by $\Phi:\cL (\cH)\to \cL (\cH)$ or $\Phi:\M{d}\to \M{d}$ and use $\St{d}$ to denote the set of quantum states $\cD (\cH)$.
\end{rem}

\subsection{Spectral and ergodic properties}\label{sec;spec}

Every quantum channel $\Phi: \cL({\cH})\to \cL(\cH)$ admits a quantum state $\rho\in \cD(\cH)$ as a fixed point: $\Phi(\rho)=\rho$ \cite[Theorem 4.24]{watrous2018theory}. In other words, $\lambda=1$ is always an eigenvalue of $\Phi$. The spectrum  (denoted $\operatorname{spec}\Phi$) of $\Phi$ is contained within the unit disk $\{ z\in \mathbb{C}: |z|\leq 1\}$ in the complex plane and is invariant under complex conjugation, i.e., $\lambda\in \operatorname{spec}\Phi \implies \overbar \lambda \in \operatorname{spec}\Phi$. The peripheral spectrum of $\Phi$ consists of all peripheral eigenvalues $\lambda \in \mathbb{T} \cap \operatorname{spec}\Phi$, where $\mathbb{T}:= \{z\in \mathbb{C} : |z|=1 \}$. It is known that the geometric and algebraic multiplicities of all peripheral eigenvalues of a quantum channel are equal \cite[Proposition 6.2]{Wolf2012Qtour}. A peripheral eigenvalue is called simple if it has unit multiplicity. A quantum channel and its adjoint both share the same spectrum, $\operatorname{spec}\Phi = \operatorname{spec}\Phi^*$.

We now introduce the notions of ergodic and mixing quantum channels. For a more detailed study of the ergodic theory of quantum channels, the readers should refer to \cite{Burgarth2013ergodic, Wolf2012Qtour, Singh2022ergodic}

\begin{theorem}\label{theorem:ergodic}
For a quantum channel $\Phi: \cL({\cH})\to \cL(\cH)$,
%$\Phi:\M{d}\to \M{d}$, 
the following are equivalent.
\begin{itemize}
    \item $\lambda=1$ is a simple eigenvalue of $\Phi$.
    \item There exists a state $\rho\in \cD(\cH)$ such that for all $X\in \cL(\cH)$,
    \begin{equation*}
        \lim_{N\to \infty} \frac{1}{N} \sum_{n=0}^{N-1} \Phi^n(X) = \operatorname{Tr}(X)\rho.
    \end{equation*}
\end{itemize}
A channel $\Phi$ (or a dQMS $\{\Phi^n\}_{n\in \mathbb N}$) satisfying these equivalent conditions is said to be \emph{ergodic}. The unique fixed point $\rho\in \cD(\cH)$ of $\Phi$ is called the \emph{invariant state} of $\Phi$. If, in addition, the unique invariant state has full rank, then the channel and the associated dQMS are said to be \emph{irreducible}.
\end{theorem}

\begin{theorem}\label{theorem:mixing}
For a quantum channel $\Phi: \cL({\cH})\to \cL(\cH)$,
%$\Phi:\M{d}\to \M{d}$, 
the following are equivalent.
\begin{enumerate}
    \item $\lambda=1$ is a simple eigenvalue of $\Phi$ and there are no other peripheral eigenvalues of $\Phi$.
    \item There exists a state $\rho\in
    \cD(\cH)$ such that for all $X\in \cL(\cH)$, 
    \begin{equation*}
    \lim_{n\to \infty} \Phi^n(X) = \operatorname{Tr}(X)\rho.    
    \end{equation*}
\end{enumerate}
A channel $\Phi$ (or a dQMS $\{\Phi^n\}_{n\in \mathbb N}$) satisfying these equivalent conditions is said to be \emph{mixing}. If, in addition, the unique invariant state of $\Phi$ has full rank, then the channel and the associated dQMS are said to be \emph{primitive}.
\end{theorem}

The peripheral space of a channel $\Phi:\B{\Hil}\to \B{\Hil}$ is defined as the span of all its peripheral eigenoperators:
\begin{equation}
    \chi(\Phi):= \operatorname{span} \{X\in \B{\Hil} : \Phi (X) = \lambda X, \abs{\lambda}=1 \}.
\end{equation}

For any channel $\Phi:\B{\Hil}\to \B{\Hil}$, there exists a decomposition $\Hil = \Hil_0 \oplus \bigoplus_{k=1}^K \Hil_{k,1}\otimes \Hil_{k,2}$ and positive definite states $\rho_{k}\in \B{\Hil_{k,2}}$ such that \cite[Theorem 6.16]{Wolf2012Qtour}:
    \begin{equation}\label{eq:phasespace}
        \chi(\Phi) = 0 \oplus \bigoplus_{k=1}^K (\B{\Hil_{k,1}}\otimes \rho_k). 
    \end{equation}
Moreover, there exist unitaries $U_k\in \B{\Hil_{k,1}}$ and a permutation $\pi$ which permutes within subsets of $\{1,2,\ldots ,K \}$ for which the corresponding $\Hil_{k,1}$'s have the same dimension, such that for any
    \begin{equation}\label{eq:phaseaction}
    X = 0 \oplus \bigoplus_{k=1}^K x_k \otimes \rho_k, \quad    \Phi (X) = 0 \oplus \bigoplus_{k=1}^K U^{\dagger}_k x_{\pi (k)} U_k \otimes \rho_k .
    \end{equation}

A channel $\Phi:\B{\Hil_A}\to \B{\Hil_B}$ is called \emph{entanglement-breaking} if local action of $\Phi$ on any bipartite system breaks all entanglement in the system, i.e., for all states $\rho\in \State{\Hil_A \otimes \Hil_R}$, $(\Phi \otimes \id_R) (\rho)$ is separable, where $R$ is an arbitrary reference system.

A channel $\Phi:\B{\Hil}\to \B{\Hil}$ (or a dQMS $\{\Phi^n \}_{n\in \mathbb{N}}$) is called 
\begin{itemize}
    \item \emph{eventually} entanglement-breaking if there exists an $n\in \mathbb{N}$ such that $\Phi^n$ is entanglement-breaking. 
    \item \emph{asymptotically} entanglement-breaking if all the limit points of the set $\{\Phi^n \}_{n\in \mathbb{N}}$ are entanglement-breaking.
\end{itemize}

It is known that for any channel $\Phi:\B{\Hil}\to \B{\Hil}$, the limit points of the set $\{\Phi^n \}_{n\in \mathbb{N}}$ are either all entanglement-breaking or none of them are \cite{Lami2016eb}. Moreover, the following result was derived in \cite[Theorem 32]{Lami2016eb}.

\begin{theorem}
    Let $\Phi:\B{\Hil}\to \B{\Hil}$ be a quantum channel. The following are equivalent:
    \begin{itemize}
        \item $\Phi$ is asymptotically entanglement-breaking.
        \item All the $\B{\Hil_{k,1}}$ blocks in the peripheral decomposition in eq.~\eqref{eq:phasespace} are one-dimensional. 
        \item All peripheral points of $\Phi$ communte with one another, i.e. $\forall X,Y\in \chi(\Phi)$, $[X,Y]=XY - YX = 0$.
    \end{itemize}
\end{theorem}

For an elaborate discussion of eventually entanglement-breaking and asymptotically entanglement-breaking quantum channels, the readers should refer to \cite{Lami2016eb, eric2020eb}.

\subsection{Fixed points and multiplicative domains}\label{subsec-fixed-pt-mult.} 

Given a channel $\Phi$ on $\M{d}$, the set of fixed points of $\Phi$ is the set 
\[\mathrm{Fix}_\Phi=\{A\in \M{d}| \ \Phi(A)=A\}.\]
Note that it is a vector subspace of $\M{d}$ which is closed under taking adjoints. If the channel $\Phi$ is unital, then the set $\mathrm{Fix}_\Phi$ is also closed under multiplication and hence is a C$^*$-subalgebra (\cite{kribs}). Recall that a (operator) norm-closed subset of $\cL (\Hil)$ that is closed under addition, multiplication, and the $*$-operation is called a C$^*$-algebra.  

The multiplicative domain of $\Phi$ is defined to be the following set
\[\mathcal{M}_\Phi=\{A\in \M{d}| \ \Phi(AX)=\Phi (A)\Phi (X), \ \Phi(XA)=\Phi(X)\Phi(A), \forall X\in \M{d}\}.\]
%In the case where the channel is unital, the set $\mathrm{Fix}_\Phi$ is also closed under multiplication and hence is a C$^*$-subalgebra. Recall that a (operator) norm-closed subset of $\cL (\Hil)$ that is closed under addition, multiplication, and the $*$-operation is called a C$^*$-algebra.  

 For unital channels, it holds that 
\[\mathrm{Fix}_\Phi=\mathcal{A}',\]
where $\mathcal{A}$ is the C$^*$-algebra generated by the Kraus operators of $\Phi$ and $\mathcal{A}'$ is the algebra that commutes with $\mathcal{A}$. Also, it is known that for unital channels one has
\[\mathcal{M}_\Phi=\mathrm{Fix}_{\Phi^*\circ\Phi}.\]
%where {\nilan{$\Phi^*$ is the adjoint channel of $\Phi$ (with respect to the Hilbert-Schmidt inner product).}}
The multiplicative domain shrinks under the iterations of a unital channel. Indeed, let $\mathcal{M}_{\Phi^n}$ denote %the corresponding 
the multiplicative domain of $\Phi^n$ for each $n\in \mathbb{N}$, then it holds that (\cite{rahaman2017})
  $$ \mathcal{M}_{\Phi}\supseteq \mathcal{M}_{\Phi^2} \supseteq \cdots \mathcal{M}_{\Phi^n} \supseteq\cdots.$$
 The above chain stabilizes at a set which we denote as $\mathcal{M}_{\Phi^{\infty}}:=\bigcap_{n\in 
\mathbb{N}} \mathcal{M}_{\Phi^n}$ and call the {\em{stabilized multiplicative domain}} of $\Phi$ and it is invariant under repeated applications of the channel.

\begin{comment}
    \begin{definition}\label{def:ergodic-mixing}
A quantum channel $\Phi :\M{d} \to \M{d}$ is called 
\begin{enumerate}
\item
{\em{ergodic}} if the $\Phi=1$ eigenvalue of $\Phi$ is simple.
\item 
{\em{mixing}} if it is ergodic and $\operatorname{spec}\, \Phi \cap \mathbb{T} = \{1\}$. {\color{green}(This is equivalent to saying that $\forall \rho$, $\lim_{n \to \infty} \Phi^n(\rho) = \Tr(\rho)\rhos$.)}  
\end{enumerate}
\end{definition}
\end{comment}

\subsection{Contraction coefficient}
The contraction coefficient of a quantum channel $\Phi:\cL (\Hil_A)\to \cL (\Hil_B)$ with respect to the trace norm is defined as follows \cite{Hiai2015contraction}:
\begin{align}\label{eq:coeff}
    \eta^{\Tr}(\Phi) := \sup_{\substack{\rho,\sigma\in \State{\Hil_A} \\ \rho\neq \sigma}} \frac{\norm{\Phi(\rho) - \Phi(\sigma)}_1}{\norm{\rho - \sigma}_1}.
    \end{align}
\begin{lemma}\label{lemma:supTr} \cite[Theorem 2]{Ruskai1994contraction}
For a quantum channel $\Phi:\cL (\Hil_A)\to \cL (\Hil_B)$,
\begin{equation*}
    \eta^{\Tr}(\Phi) = \sup_{\substack{ \rho,\sigma\in \State{\Hil_A} \\ \rho \perp \sigma}} \frac{1}{2} \norm{\Phi(\rho) - \Phi(\sigma)}_1.
\end{equation*}
Moreover, the states in the supremum above can be taken to be pure.
\end{lemma}

\subsection{Zero error communication}\label{sec:zero}

Let Alice and Bob be linked via a quantum channel $\Phi: \B{\Hil_A}\to \B{\Hil_B}$.
Suppose Alice wants to communicate  $M\geq 2$ classical messages to Bob perfectly without error. This is possible if and only if she can encode the $M$ messages in states 
\begin{equation}
    \{\rho_m\}_{m=1}^M\subset \State{\Hil_A} \quad \text{such that}\quad \forall m\neq m': \quad \Phi(\rho_m) \perp \Phi(\rho_{m'}),
\end{equation}
where we say that two positive operators $X,Y \geq 0$ are orthogonal ($X\perp Y$) if their supports are orthogonal as subspaces, which is equivalent to saying that $\Tr (XY) = 0$. The interpretation here is that for any choice of encoding of the $M$ classical messages on the input side, the set of output states would have to be perfectly distinguishable in order for Bob to decode the intended classical message via a measurement without error, which is possible if and only if the output states are pairwise orthogonal. Note that if $\rho\in \State{\Hil_A}$ is a mixed state and $\ket{\psi}\in \supp (\rho)$, there exists an $\epsilon>0$ such that $\psi \leq \epsilon \rho$. Therefore, $ \Phi(\psi) \leq \epsilon \Phi(\rho)$ and $\supp (\Phi(\psi) )\subseteq \supp (\Phi(\rho))$. Thus, without loss of generality, all the encoding states in the above scheme can be taken to be pure. With this background, we can introduce the following definition.

\begin{definition} \label{def:zero-classical}
The one-shot zero-error classical capacity $C^{(1)}_0(\Phi)$ of a channel $\Phi: \B{\Hil_A}\to \B{\Hil_B}$ is defined as follows:
\begin{align}
    C_0^{(1)}(\Phi) :=\sup_{C}\log |C|,
\end{align} 
where the supremum is over all collections $C$ of pure quantum states $\{\psi_m\}_{m=1}^{|C|} \subset \cD({\Hil_A})$ such that
\begin{equation}
\forall \, m\neq m': \quad  \Phi(\psi_m) \perp \Phi(\psi_{m'}).
\end{equation}
\end{definition}

Suppose now that Alice and Bob share some entanglement beforehand, say in the form of a pure bipartite state $\psi \in \State{\Hil_{A_0} \otimes \Hil_{B_0}}$. Alice can now come up with a more general encoding scheme by pre-processing her share of $\psi$ with arbitrary quantum channels $\{\mathcal{E}_m : \B{\Hil_{A_0}}\to \B{\Hil_A}\}_{m=1}^M$. She then sends her share of the resulting states through $\Phi : \B{\Hil_A}\to \B{\Hil_B}$. As before, the condition for perfect distinguishability on Bob's end is equivalent to the following orthogonality relations:
\begin{equation}
    \forall m\neq m': \quad (\Phi \circ \mathcal{E}_m \otimes \mathrm{id}_{B_0})(\psi) \perp  (\Phi \circ \mathcal{E}_{m'} \otimes \mathrm{id}_{B_0})(\psi).
\end{equation}

\begin{definition}\label{def:zero-ent-classical}
The one-shot zero-error entanglement assissted classical capacity $C^{(1)}_{0E}(\Phi)$ of a quantum channel $\Phi: \B{\Hil_A}\to \B{\Hil_B}$ is defined as follows:
\begin{align}
    C_{0E}^{(1)}(\Phi) :=\sup_{\psi, C}\log |C|,
\end{align} 
where the supremum is over all pure bipartite states $\psi\in \State{\Hil_{A_0} \otimes \Hil_{B_0}}$ and collections $C$ of quantum channels $\{\mathcal{E}_m : \B{\Hil_{A_0}}\to \B{\Hil_A} \}_{m=1}^{|C|}$ such that
\begin{equation}
    \forall m\neq m': \quad (\Phi \circ \mathcal{E}_m \otimes \mathrm{id}_{B_0})(\psi) \perp  (\Phi \circ \mathcal{E}_{m'} \otimes \mathrm{id}_{B_0})(\psi).
\end{equation}
\end{definition}

If Alice wants to send quantum information to Bob through $\Phi:\B{\Hil_A}\to \B{\Hil_B}$ perfectly without error, she must find an encoding subspace $\cC\subseteq \Hil_A$ such that Bob can reverse the action of $\Phi$ on $\cC$, i.e., there exists a recovery channel $\mathcal{R}:\B{\Hil_B}\to \B{\Hil_A}$ such that for all states $\rho\in \State{\Hil_A}$ with $\supp (\rho) \subseteq \mathcal{C}$\footnote{Note that the support of a positive semi-definite operator is precisely the orthogonal complement of its kernel.},
\begin{equation}
    \mathcal{R}\circ \Phi (\rho) = \rho. 
\end{equation}
If $\{K_i\}_i$ are the Kraus operators of $\Phi$, the Knill-Laflamme error correction conditions \cite{knil-laf} show that a subspace $\cC\subseteq \Hil_A$ as above exists if and only if $P_{\cC} K^{*}_i K_j P_{\cC} = \lambda_{ij} P_{\cC}$ for all $i,j$, where $P_{\cC}$ denotes the orthogonal projection onto $\cC$ and $\lambda_{ij}\in \mathbb{C}$ are complex numbers.

\begin{definition}\label{def:zero-quantum}
    The one shot zero-error quantum capacity $Q_0^{(1)}(\Phi)$ of a channel $\Phi:\B{\Hil_A}\to \B{\Hil_B}$ is defined as follows:
\begin{align}
Q_0^{(1)}(\Phi) &:= \sup_{\cC} \log \dim(\cC),
\end{align}
where the supremum is over all subspaces $\cC\subseteq \Hil_A$ for which there exists a recovery quantum channel $\mathcal{R}:\B{\Hil_B}\to \B{\Hil_A}$ satisfying $\mathcal{R}\circ \Phi (\rho) = \rho$ for all states $\rho\in \State{\Hil_A}$ with $\supp (\rho ) \subseteq \mathcal{C}$. 
\end{definition}

It is possible to recast the above channel capacity definitions in terms of an operator system that one can associate with the channel. 
\begin{definition}\label{def: op-sys}
    Let $\Phi: \B{\Hil_A}\to \B{\Hil_B}$ have a Kraus representation $\Phi(X)=\sum_{i=1}^p K_i XK_i^*$. The operator system (also called the \emph{non-commutative (confusability) graph}) of $\Phi$ is defined as~\cite{Duan2013noncomm}
\begin{equation}
    S_{\Phi} := {\rm{span}} \{K^*_i K_j, \, 1\leq i,j \leq p\} \subseteq \B{\Hil_A}. 
\end{equation}
\end{definition}

One can check that the above definition is independent of the chosen Kraus representation of $\Phi$. Moreover, $\sum_{i=1}^p K_i^* K_i={\iden}_A \in S_{\Phi}$, since $\Phi$ is trace-preserving. Furthermore, $X\in S_{\Phi} \implies X^{*} \in S_{\Phi}$. Such $*-$closed subspaces $S\subseteq \B{\Hil}$ containing the identity are called \emph{operator systems} \cite{paulsen-book}. For an operator system $S\subseteq \B{\cH}$,
\begin{itemize}
    \item the maximum size $M$ of a set of mutually orthogonal vectors $\{\ket{\psi_m} \}_{m=1}^M \subseteq \cH$ such that 
\begin{equation}
    \forall m\neq m': \quad |\psi_m\rangle \langle\psi_{m'} | \perp S,
\end{equation}
is called the \emph{independence number} of $S$ (denoted $\alpha(S)$). %The set $\{\ket{\psi_m} \}_{m=1}^k \subseteq \cH$ is called an \emph{independent set} for $S$.
\item the maximum number $M$ such that there exist Hilbert spaces $\cH_{A_0}, \cH_R$, a state $\rho\in \State{\Hil_{A_0}}$, and isometries $\{V_m: \Hil_{A_0} \to \Hil\otimes \Hil_R\}_{m=1}^M$ such that 
\begin{equation}
\forall m\neq m': \quad V_m \rho V_{m'} \perp S\otimes \B{\Hil_R},
\end{equation}
is called the \emph{entanglement-assisted independence number} of $S$ (denoted $\Tilde{\alpha}(S)$).  
\item the maximum number $M$ such that there exists a subspace $\cC\subseteq \Hil$ with $\dim \cC=M$ satisfying $P_\cC S P_\cC = \mathbb{C} P_\cC$,
(where $P_\cC$ denotes the orthogonal projection onto $\mathcal{C}$)
is called the \emph{quantum independence number} of $S$ (denoted $\alpha_q(S)$).  
\end{itemize}

\begin{theorem}\cite{Duan2013noncomm}\label{op-syst-char}
    For any channel $\Phi: \B{\Hil_A}\to \B{\Hil_B}$, the following relations hold:
\begin{align}
    C^{(1)}_0 (\Phi) &= \log \alpha (S_{\Phi}) \\ 
    C^{(1)}_{0E} (\Phi) &= \log \Tilde{\alpha} (S_{\Phi}) \\
    Q^{(1)}_0 (\Phi) &= \log \alpha_q (S_{\Phi})
\end{align}
Moreover, $0\leq Q_0^{(1)}(\Phi)\leq C_0^{(1)}(\Phi)\leq C_{0E}^{(1)}(\Phi)$, where the inequalities can be strict.
\end{theorem}

\begin{remark}
    The terminology in Definition~\ref{def: op-sys} is motivated by the notion of confusability graphs of classical channels \cite{Shannon1956zero}. A discrete classical channel $\cN: \cX \to \cY$, where $\cX$ and $\cY$ denote two finite alphabets, is defined by a transition probability matrix, $A$, with elements $\cN(y|x)$ that express the probability of observing the symbol $y$ given that the symbol $x$ was sent. In order to send different messages through the channel $\cN$ with zero error, they should be encoded in the symbols of $\cX$ in a manner such that the corresponding outputs of the channel have disjoint support. One can associate a {\em{confusability graph}} $G_\cN$ with the channel; it has vertex set $\cX$ and edges between any pair $x, x'\in \cX$ which can be confused, i.e.~for which there is a $y\in \cY$ such that $\cN(y|x)\cN(y|x') >0$. The one-shot zero error capacity of $\cN$ is the maximum number of bits of classical information that can be transmitted without error through a single use of $\cN$. This is given by $\log \alpha(G_\cN)$, where $\alpha(G_\cN)$ is the {\em{independence number}} of the confusability graph, and is equal to the maximum number of vertices in $G_\cN$ which do not have any edges between them.
\end{remark}

We now collect some results from the literature which describe equivalent conditions for the various zero-error one-shot capacities of a channel to be zero. Let us first introduce the following terminology. A channel $\Phi:\B{\Hil_A}\to \B{\Hil_B}$ is called 
\begin{itemize}
    \item \emph{c-scrambling} if $C_0^{(1)}(\Phi)=0$.
    \item \emph{q-scrambling} if $Q_0^{(1)}(\Phi)=0$.
\end{itemize}

\begin{theorem}\label{thm:scrambling}
    The following are equivalent for a quantum channel $\Phi: \cL (\Hil_A)\to \cL (\Hil_B)$:
    \begin{enumerate}
        \item $\Phi$ is c-scrambling.
        \item $\eta^{\Tr}(\Phi)<1$.
        \item $\Tr (\Phi (\psi) \Phi (\phi))>0$ for any pair of orthogonal pure states $\psi, \phi\in \State{\Hil_A}$.
        \item $\Tr (\Phi (A) \Phi (B))>0$ for all non-zero positive operators $A,B\in \cL (\Hil_A)$.
        \item There are no rank one elements in $S_{\Phi}^{\perp}$.
    \end{enumerate}
\end{theorem}
The equivalence of $(1),(2),(3)$ and $(4)$ above was obtained in \cite[Proposition 4.2]{Hiai2015contraction} and the equivalence of $(1)$ and $(5)$ was obtained in \cite{Duan2009zero}.

\begin{comment}
\begin{proof}
    $(1)\implies (4)$ Clearly, if $C^{(1)}_0(\Phi)=0$, then for all $\rho,\sigma\in \State{\Hil_A}$, we have $\Tr(\Phi(\rho)\Phi(\sigma))>0$. For any non-zero positive operators $A,B$, the conclusion in $(4)$ follows by choosing $\rho=A/\Tr A$ and $\sigma = B/ \Tr B$ above. $(4)\implies (3)$ Trivial. $(3)\implies (2)$ This follows immediately from Lemma~\ref{lemma:supTr} and the fact that for any two states $\rho,\sigma\in \State{\Hil_A}$, we have $\norm{\rho - \sigma}_1 \leq 2$ and $\norm{\rho -\sigma}_1 = 2 \iff \Tr (\rho \sigma)=0$. $(2)\implies (1)$ If $\Phi$ is scrambling, then for all $\rho\neq \sigma\in \State{\Hil_A}$, we get  $$\norm{\Phi (\rho) - \Phi (\sigma)}_1 < \norm{\rho - \sigma}_1 \leq 2.$$ In other words, for all $\rho\neq \sigma \in \State{\Hil_A}$, we get $\Tr (\Phi(\rho)\Phi (\sigma))>0$. Hence, $C^{(1)}_0(\Phi)=0$. \\
    Finally, the equivalence of $(1)$ and $(5)$ follows immediately from the relation $C^{(1)}_0(\Phi) = \log \alpha (S_{\Phi})$.
\end{proof}
\end{comment}

\begin{proposition}\cite{Duan2013noncomm}\label{prop:C0E=0}
    For a quantum channel $\Phi: \cL (\Hil_A)\to \cL (\Hil_B)$, $C_{0E}^{(1)}(\Phi)=0$ if and only if $S_{\Phi}^{\perp}$ is the zero subspace.
\end{proposition}

\begin{proposition}\label{prop:C0Q0=0} \cite{Shirokov-Shulman2103}
    Let $\Phi: \cL (\Hil_A)\to \cL (\Hil_B)$ be a quantum channel. Then, $Q_0^{(1)}(\Phi)>0$ (i.e. $\Phi$ is not q-scrambling) if and only if there are unit vectors $\ket{\xi}, \ket{\eta}\in \Hil_A$ such that
\[\forall X\in S_\Phi: \quad \langle \xi | X |\eta\rangle=0 \quad \text{and} \quad  \langle \xi | X | \xi\rangle=\langle \eta | X |\eta\rangle.\] Moreover, the following implications hold:
\begin{itemize}
\item  $[S_\Phi]'$ is non-abelian $\implies Q_0^{(1)}(\Phi)>0$.
\item $[S_\Phi]'$ is non-trivial $\implies C_0^{(1)}(\Phi)>0$.
\end{itemize}
If $S_\Phi$ is an algebra, the reverse implications also hold:
\begin{itemize}
\item  $[S_\Phi]'$ is non-abelian $\iff Q_0^{(1)}(\Phi)>0$.
\item $[S_\Phi]'$ is non-trivial $\iff C_0^{(1)}(\Phi)>0$.
\end{itemize}
In the above statements we used the notation $[S_\Phi]'$ to denote the commutant of $[S_\Phi]$, i.e.,
\[[S_\Phi]'=\{X\in \cL (\Hil_A): X Y= Y X, \,\forall Y\in S_\Phi\}.\]
\end{proposition}

\begin{comment}
    \begin{proof}
For any distinct $\rho,\sigma\in \St{d}$, $A = \rho - \sigma $ is a Hermitian operator, and thus can be decomposed into its positive and negative parts: $A= A^+ - A^-$, where $A^+\perp A^-$ and $\norm{\rho - \sigma}_1 = \Tr A^+ + \Tr A^-$. Note that since $\rho,\sigma$ are distinct, both $A^+$ and $A^-$ are non-zero. Moreover, $\Tr A^+ = \Tr A^-$, since $\Tr A = 0$. Hence,
\begin{equation}
    \frac{\norm{\Phi (\rho) - \Phi (\sigma) }_1}{\norm{\rho - \sigma}_1} = \frac{1}{2} \norm{ \frac{\Phi (A^+)}{\Tr A^+} - \frac{\Phi (A^-)}{\Tr A^-}  }_1.
\end{equation}
This shows that 
\begin{equation}
    \sup_{\substack{\rho,\sigma\in \St{d} \\ \rho\neq \sigma}} \frac{\norm{\Phi(\rho) - \Phi(\sigma)}_1}{\norm{\rho - \sigma}_1} = \sup_{\substack{ \rho,\sigma\in \St{d} \\ \rho \perp \sigma}} \frac{1}{2} \norm{\Phi(\rho) - \Phi(\sigma)}_1.
\end{equation}
To see why the supremum can be restricted to pure orthogonal states, we can spectrally decompose $\rho,\sigma\in \St{d}$ as $\rho = \sum_i \Phi_i \ketbra{\Phi_i}$ and $\sigma = \sum_i \mu_i \ketbra{\mu_i}$, and note that 
\begin{align}
    \norm{\Phi (\rho) - \Phi (\sigma)}_1 &= \norm{ \sum_i \Phi_i \Phi (\ketbra{\Phi_i}) - \sum_j \mu_j \Phi (\ketbra{\mu_j}) }_1 \\ 
    &\leq \sum_{i,j} \Phi_i \mu_j \norm{\Phi(\ketbra{\Phi_i}) - \Phi(\ketbra{\mu_j}) }_1 \\
    &\leq \sup_{i,j} \norm{\Phi(\ketbra{\Phi_i}) - \Phi(\ketbra{\mu_j}) }_1.
\end{align}
\end{proof}
\end{comment}

\section{Main results}\label{sec:main}

\subsection{Characterization of all eventually c-scrambling dQMS}

Consider an open quantum system $A$ with associated Hilbert space $\Hil_A\simeq \Hil \simeq \C{d}$ whose time evolution is governed by a dQMS $\{\Phi^n \}_{n\in \mathbb N}$, where $\Phi:\B{\Hil}\to \B{\Hil}$ is a quantum channel. We call  $\{\Phi^n \}_{n\in \mathbb N}$ \emph{eventually c-scrambling} if there exists an $n\in \mathbb N$ such that $\Phi^n$ is c-scrambling. These are precisely the kind of evolutions which eventually become useless for zero-error classical communication. Since any non-trivial $\Phi$ models some inherent noise in the system, one might naively reason that any (non-trivial) dQMS is eventually scrambling, i.e., if one waits for a long enough time $n\in \mathbb N$, $\Phi^n$ will become too noisy to communicate any classical message perfectly. However, the following theorem shows that this is the case only for the class of mixing evolutions. Moreover, we prove that if a dQMS is eventually scrambling, then it will eventually become useless for classical communication even if entanglement is present to aid the process.

%Naively, since repeated applications of a noisy channel on the system obviously increases the amount of noise in the system. It is then reasonable to expect that for any channel $\Phi:\cL (\cH) \to \cL (\cH)$, the channel 
%\begin{equation}
%\Phi^k = \underbrace{\Phi \circ \Phi \circ \ldots \circ \Phi}_{k\,  \text{times}},    
%\end{equation}
%for a large enough $k\in \mathbb{N}$ becomes too noisy to communicate any classical messages perfectly (over a single use of $\Phi^k$). However, the following theorem shows that this happens only for the class of mixing channels. Note that the set of mixing channels form a dense subset of the set of quantum channels (see Section 4 in \cite{Burgarth2013ergodic}).

%\satvik{Add remark about the fact that mixing channels are dense.}

\begin{theorem}\label{theorem:scr-mix}
    Let $\Phi: \B{\Hil}\to \B{\Hil}$ be a channel and $\{\Phi^n \}_{n\in \mathbb N}$ be the associated dQMS. Then, the following are equivalent.
    \begin{enumerate}
        \item $\exists k\in \mathbb{N}$ such that $C^{(1)}_{0E}(\Phi^k)=0$.
        \item $\exists k\in \mathbb{N}$ such that $C^{(1)}_{0}(\Phi^k)=0$, i.e., $\{\Phi^n \}_{n\in \mathbb N}$ is eventually c-scrambling.
        \item $\{\Phi^n \}_{n\in \mathbb N}$ is mixing.
    \end{enumerate}
\end{theorem}

\begin{proof}
$(1)\implies (2)$ This implication is trivial, since $C^{(1)}_0(\Phi)\leq C^{(1)}_{0E}(\Phi)$ for any channel $\Phi$. \\ $(2)\implies (3)$ Assume that $\exists k\in \mathbb{N}$ such that $\Phi^k$ is scrambling, i.e., $\eta^{\Tr}(\Phi^k)=c<1$. Let $\rho\in \State{\cH}$ be a fixed state of $\Phi$. Then, for any (non-zero) positive semi-definite operator $X\in \B{\cH}$, we have
    \begin{equation}
        \norm{\Phi^{nk}(X/\Tr X)-\rho}_1 = \norm{\Phi^{nk}(X/\Tr X)-\Phi^{nk}(\rho)}_1 \leq c^n \norm{X/\Tr X - \rho}_1 \to 0 \text{ as } n\to \infty.
    \end{equation}
    Hence, for all positive semi-definite $X\in \B{\cH}$, we get $\lim_{n\to \infty} \Phi^n (X) = \Tr (X) \rho$. Since any $X\in \B{\cH}$ can be written as a linear combination of positive semi-definite operators, it is clear that $\Phi$ is mixing. \\
    $(3)\implies (1)$ Assume that $\exists \rho\in\State{\cH}$ such that $\forall X\in \B{\cH}$, $\lim_{n\to \infty} \Phi^n(X) = \operatorname{Tr}(X)\rho$. Since pointwise and uniform convergence are equivalent in finite dimensions, for every $\epsilon>0$, $\exists N\in \mathbb{N}$ such that $\forall X\in \B{\cH}$, $\norm{\Phi^n (X) - \Tr (X)\rho}_1 \leq \epsilon$ for $n\geq N$. In other words, $\lim_{n\to \infty} \Phi^n = \Phi_{\infty}$, where $\Phi_{\infty}$ is the completely depolarizing channel defined as $\Phi_{\infty}(X) = \Tr (X)\rho$, and the convergence is with respect to the induced trace norm defined as
    \begin{equation}
        \norm{\Phi}_1 := \sup _{\norm{X}_1\leq 1} \norm{\Phi (X)}_1.
    \end{equation}
    (In fact, since any two norms on a finite-dimensional space are equivalent, we can also think of this convergence in terms any other norm, say the diamond norm for instance.) It is then easy to see that
    $\lim_{n\to \infty} \Phi^n \otimes \mathrm{id}_{B_0} = \Phi_{\infty}\otimes \mathrm{id}_{B_0}$ for any auxiliary system $B_0$. Put differently, for every $\epsilon>0$, $\exists N\in \mathbb{N}$ such that $\forall X\in \B{\cH\otimes \cH_{B_0}}$, $\norm{(\Phi^n \otimes \mathrm{id}_{B_0}) (X) - (\Phi_{\infty}\otimes \mathrm{id}_{B_0})(X)}_1 \leq \epsilon$ for $n\geq N$. Hence, for any pure state $\psi\in \State{\Hil_{A_0} \otimes \Hil_{B_0}}$, channels $\mathcal{E}_i : \B{\Hil_{A_0}}\to \B{\cH}$ for $i=1,2$, and $n\geq N$, we have that
    \begin{align}
        &\frac{1}{2}\norm{ (\Phi^n \circ \mathcal{E}_1 \otimes \mathrm{id}_{B_0}) (\psi) - (\Phi^n \circ \mathcal{E}_2 \otimes \mathrm{id}_{B_0}) (\psi)}_1 \\ 
        =\,\, &\frac{1}{2}\norm{ (\Phi^n \otimes \mathrm{id}_{B_0}) (\psi_1) - (\Phi^n \otimes \mathrm{id}_{B_0}) (\psi_2)}_1 \\ 
        \leq \,\, &\frac{1}{2}\norm{ (\Phi^n \otimes \mathrm{id}_{B_0}) (\psi_1) - (\Phi_{\infty}\otimes \mathrm{id}_{B_0}) (\psi_1)}_1 + \frac{1}{2}\norm{(\Phi^n \otimes \mathrm{id}_{B_0})(\psi_2) - (\Phi_{\infty}\otimes \mathrm{id}_{B_0}) (\psi_2)}_1 \leq \epsilon
    \end{align}
    where $\psi_i = (\mathcal{E}_i \otimes \mathrm{id}_{B_0}) (\psi)$ for $i=1,2$ and $(\Phi_{\infty}\otimes \mathrm{id}_{B_0}) (\psi_1) = (\Phi_{\infty}\otimes \mathrm{id}_{B_0}) (\psi_2)$ were added and subtracted to obtain the second inequality. By letting $\epsilon$ be sufficiently small, this implies that $C^{(1)}_{0E}(\Phi^k)=0$ for some $k\in \mathbb{N}$.
\end{proof}

\subsection{Zero-error classical encodings for non-mixing dQMS}

Let us now discuss the conclusion of Theorem~\ref{theorem:scr-mix} in more detail. The theorem shows that the semigroups $\{\Phi^n\}_{n\in \mathbb N}$ that are able to send classical messages perfectly for arbitrarily long times are precisely of the non-mixing type:
\begin{equation}
    \forall n\in \mathbb{N}: \quad C^{(1)}_0(\Phi^n) > 0 \iff \{\Phi^n\}_{n\in \mathbb N} \text{ is non-mixing.}
\end{equation}
It is then natural to ask what kind of encoding states $\rho_1,\rho_2\in \M{d}$ can be used to perfectly transmit a 1-bit classical message through $\Phi^n$ ($\forall n\in \mathbb N$) for a given non-mixing dQMS $\{\Phi^n \}_{n\in \mathbb N}$. Since $\{\Phi^n \}_{n\in \mathbb N}$ is non-mixing, the following two cases can arise:  \\

%{\nilan{ND: I put the phrase above in green to point out that we doqt address this question in Case II. Rather, in the latter case, we ask how many messages can be transmitted perfectly and using what encodings? Right?}} 

\noindent \textbf{Case I.} $\lambda=1$ is not a simple eigenvalue of $\Phi$, i.e., $\Phi$ is not ergodic. To tackle this case, let us first note a lemma.

\begin{lemma} \label{lemma:fixed}
 \cite[Section 3.1]{Burgarth2013ergodic} \cite[Proposition 6.8]{Wolf2012Qtour}   The fixed-point vector space $\mathrm{Fix}_{\Phi}:= \{A\in \M{d} : \Phi(A) = A \}$ of a quantum channel $\Phi:\M{d}\to \M{d}$ is spanned by quantum states. 
\end{lemma}
\begin{proof}
    Let $A\in \mathrm{Fix}_{\Phi}$ and consider its canonical decomposition into Hermitian parts:
    \begin{equation}
        A = \frac{A+A^{\dagger}}{2} + i \frac{A-A^{\dagger}}{2i}.
    \end{equation}
    Since $\Phi$ is Hermiticity preserving, $(A+A^{\dagger})/2$ and $(A-A^{\dagger})/2i$ are also fixed by $\Phi$. Thus, it suffices to show that a Hermitian fixed point $A=A^{\dagger}\in \mathrm{Fix}_{\Phi}$ lies in the span of quantum states. Since $A$ is Hermitian, we can write its Jordan decomposition: $A= A^+ - A^-$, where $A^\pm \geq 0$ and $\Tr (A^+A^-)=0$. Let $\Pi^+$ be the projector onto support of $A^+$. Then, we have $A^+ = \Pi^+ A = \Pi^+ \Phi (A) = \Pi^+ \Phi (A^+) - \Pi^+ \Phi (A^-)$, so that
    \begin{equation}
        \Tr (A^+) = \Tr (\Pi^+ \Phi(A^+)) - \Tr (\Pi^+ \Phi(A^-)) \leq \Tr (\Pi^+ \Phi(A^+)) \leq \Tr (\Phi(A^+)) = \Tr (A^+).
    \end{equation}
    Hence, the inequalities above must be equalities, implying that $\Pi^+ \Phi(A^-)=0$ and $\Pi^+ \Phi(A^+) = \Phi(A^+)$. This shows that $A^+ = \Pi^+ A = \Pi^+ \Phi(A^+) = \Phi(A^+)$, which clearly also implies $A^- = \Phi(A^-)$. Hence, $A$ lies in the span of quantum states $A^+/\Tr A^+$ and $A^-/\Tr A^-$.
\end{proof}

Let us now consider a channel $\Phi:\M{d}\to \M{d}$ for which $\lambda=1$ is not a simple eigenvalue. This means that there are two %{\nilan{(or more?)}} 
distinct states $\rho, \sigma\in \St{d}$ that are fixed by $\Phi$. Let $A=\rho - \sigma$.Then, the proof of Lemma~\ref{lemma:fixed} shows that the orthogonal positive and negative parts $A^\pm$ of $A$ are also fixed by $\Phi$. Thus, we can transmit a 1-bit classical message through $\Phi^n$ for all $n\in \mathbb N$ by encoding it in the states $\gamma_1 = A^+/\Tr A^+$ and $\gamma_2 = A^-/\Tr A^-$. \\

\noindent \textbf{Case II.} $\lambda=1$ is a simple eigenvalue of $\Phi$ but there are other peripheral eigenvalues as well. In other words, $\Phi$ is an ergodic quantum channel with $|\mathbb{T}\cap \text{spec}\Phi |\geq 2$. \\ 

In order to tackle this case, we need to study the peripheral spectrum of ergodic quantum channels. The structure of the peripheral spectrum of such channels is well-understood.

\begin{lemma}\label{lemma:ergodic-spectrum}
    The following is true for an ergodic quantum channel $\Phi:\M{d}\to \M{d}$:
    \begin{itemize}
        \item The peripheral spectrum of $\Phi$ is a cyclic subgroup of $\mathbb T$, i.e., $\exists q\in \mathbb N$ such that $\mathbb{T} \cap \operatorname{spec}\Phi = \{ \omega^m : m=0,1,\ldots ,q-1 \}$, where $\omega = e^{2\pi i/q}$.
        \item All the peripheral eigenvalues of $\Phi$ are simple.
        \item There exists a unitary $U\in \M{d}$ such that $\Phi^*(U^m)=\omega^m U^m$ for $m=0,1,\ldots ,q-1$.
        \item $U$ admits a spectral decomposition $U=\sum_{m=0}^{q-1} \omega^m P_m$ such that $\Phi^*(P_{m+1})=P_m$.
    \end{itemize}
\end{lemma}

\begin{proof}
    We recall that a quantum channel and its adjoint have the same spectrum, and also that if $\Phi$ is ergodic, then $\Phi^*$ is also ergodic. Also recall that the adjoint of any quantum channel is unital and completely positive. This means that for an ergodic channel $\Phi: \M{d}\to \M{d}$, the adjoint map $\Phi^*:\M{d}\to \M{d}$ has a unique positive definite fixed point: $\Phi^*(\iden) = \iden$. The peripheral spectrum of such maps has been studied in detail in the literature, and the statement of the lemma follows directly from the results in \cite{Evans1978Spectral}, see also \cite[Theorem 6.6]{Wolf2012Qtour}.
\end{proof}

The next proposition provides a lower bound on the number of classical messages that can be sent without error through any iteration of an ergodic channel.

\begin{theorem}\label{erg-per-spec} 
Let $\Phi:\M{d}\to \M{d}$ be an ergodic quantum %{\nilan{(non-mixing)}} 
channel. Then,
\begin{equation}
    \forall n\in \mathbb{N}: \quad C^{(1)}_0(\Phi^n) \geq \log |\mathbb{T} \cap \operatorname{spec}\Phi|.
\end{equation}
\end{theorem}

\begin{proof}
For an ergodic channel $\Phi:\M{d}\to \M{d}$, Lemma~\ref{lemma:ergodic-spectrum} provides orthogonal spectral projectors $P_k\in \M{d}$ satisfying $\Phi^{*}(P_{k+1})=P_k$, where $k\in \mathbb{Z}_q := \{0,1,\ldots ,q-1 \}$ and $q= |\mathbb{T} \cap \operatorname{spec}\Phi |$. Note that addition of indices here is to be understood mod $q$. We claim that $q$ classical messages can be transmitted perfectly through $\Phi^n$ for all $n\in \mathbb N$ with encoding states $\rho_m = P_m/\Tr P_m$ for $m\in \mathbb{Z}_q$. This is because
\begin{align}
   \forall m\in \mathbb{Z}_q : \quad &\Tr (\Phi^{*}(P_{m+1}) P_m) = \Tr (P_{m}) = \Tr (P_{m+1}\Phi (P_m)), \\ 
   &\Tr (\Phi^{*} (P_{m+1})P_{m'}) = \Tr (P_m P_{m'}) = 0 = \Tr (P_{m+1} \Phi (P_{m'})) \quad \text{for} \quad m\neq m'.   
\end{align}
In other words,
\begin{equation}
    \forall m\in \mathbb{Z}_q: \quad \Tr (\Phi(P_m)P_{m'}) = \begin{cases} \Tr (P_m) \quad \text{if } m'=m+1 \\
    0 \quad\quad\quad \,\,\,\, \text{if } m'\neq m+1 .
    \end{cases}
\end{equation}
This shows that sequential action of $\Phi$ on the encoding states $\rho_m$ just cyclically permutes the output supports: $\operatorname{supp} \Phi(\rho_m) \subset \operatorname{supp} P_{m+1}$ for $m\in \mathbb{Z}_q$. Clearly, the output states $\{ \Phi^n (\rho_m)\}_{m\in \mathbb{Z}_q}$ are thus mutually orthogonal for all $n\in \mathbb N$, which proves our claim.
\end{proof}

More generally, we can prove that for any dQMS $\{\Phi^n \}$, the peripheral space $\chi(\Phi)$ of the channel $\Phi$ serves as the right space to encode information in order for it to be recoverable for an arbitrarily long time. Recall that the peripheral space of a channel $\Phi:\B{\Hil}\to \B{\Hil}$ is defined as the span of all its peripheral eigenoperators and there exists a decomposition $\Hil = \Hil_0 \oplus \bigoplus_{k=1}^K \Hil_{k,1}\otimes \Hil_{k,2}$ and positive definite states $\rho_{k}\in \B{\Hil_{k,2}}$ such that \cite[Theorem 6.16]{Wolf2012Qtour}:
    \begin{equation}\label{eq:phasespace2}
        \chi(\Phi) = 0 \oplus \bigoplus_{k=1}^K (\B{\Hil_{k,1}}\otimes \rho_k). 
    \end{equation}
Moreover, there exist unitaries $U_k\in \B{\Hil_{k,1}}$ and a permutation $\pi$ which permutes within subsets of $\{1,2,\ldots ,K \}$ for which the corresponding $\Hil_{k,1}$'s have the same dimension, such that for any
    \begin{equation}\label{eq:phaseaction2}
     X = 0 \oplus \bigoplus_{k=1}^K x_k \otimes \rho_k, \quad   \Phi (X) = 0 \oplus \bigoplus_{k=1}^K U^{\dagger}_k x_{\pi (k)} U_k \otimes \rho_k .
    \end{equation}
Using this decomposition, we can prove the following result.

\begin{theorem}\label{theorem:C0phin}
    Let $\Phi:\B{\Hil}\to \B{\Hil}$ be a quantum channel and $\{\Phi^n\}_{n\in \mathbb{N}}$ be the associated dQMS. Then, $\exists N\leq (\dim \Hil) ^2$ such that
    \begin{equation}
        \forall q\in \mathbb{N}: \quad \lim_{n\to \infty} C^{(1)}_0 (\Phi^n) = C^{(1)}_0 (\Phi^N) = C^{(1)}_0 (\Phi^{N+q})=
        \log \sum_{k=1}^K \dim \Hil_{k,1}.
    \end{equation}
\end{theorem}
\begin{proof}
Note that the first two equalities follow from Theorem~\ref{stabilize}, which shows that there exists $N\leq (\dim \Hil) ^2$ such that the operator systems of the semigroup stabilize after time $N:$
\begin{equation}
    S_{\Phi} \subset S_{\Phi^2} \subset \ldots \subset S_{\Phi^N} = S_{\Phi^{N+1}} = \ldots = S_{\Phi^{N+q}} = \ldots
\end{equation}

To show the final equality, first note that the stated action of a channel on its peripheral space (Eq.~\eqref{eq:phaseaction2}) clearly implies that the set of states $\{ \ketbra{i_k} \otimes \rho_k \}$ for $k=1,2,\ldots ,K$ and $i_k = 1,2,\ldots , \dim \Hil_{k,1}$ forms a zero-error classical code for $\Phi^n$ for all $n\in \mathbb N$ in the sense of Definition~\ref{def:zero-classical}. Here, for each $k$, the state $\ketbra{i_k}\otimes \rho_k$ is supported on the $\Hil_{k,1}\otimes \Hil_{k,2}$ block in Eq.~\eqref{eq:phasespace2}. Hence,
    \begin{equation}
        \forall n\in \mathbb{N}: \quad C^{(1)}_0 (\Phi^n) \geq \log \sum_{k=1}^K \dim \Hil_{k,1}.
    \end{equation}
To show the reverse inequality, suppose that $\{\psi_m \}_{m=1}^M$ is a zero-error classical code for $\Phi^N$. Since the operator systems of the semigroup stabilize after time $N$, $\{\psi_m \}_{m=1}^M$ is also a zero-error classical code for $\Phi^{N+q}$ for all $q\in \mathbb{N}$:
\begin{equation}
    \forall q\in \mathbb N, \,\, \forall m\neq m': \quad \langle \Phi^{N+q} (\psi_m), \Phi^{N+q} (\psi_{m'}) \rangle =0.
\end{equation}
From \cite[Proposition 6.3]{Wolf2012Qtour}, there exists an increasing subsequence $(n_i)_{i\in\mathbb{N}}$ such that $\lim_{i\to \infty} \Phi^{n_i} = \mathcal{P}_{\chi}$, where $\mathcal{P}_\chi$ is the channel that projects onto the peripheral space $\chi (\Phi)$. Thus, $\{\psi_m \}_{m=1}^M$ is also a zero-error classical code for $\mathcal{P}_{\chi}$: 
\begin{equation}
    \forall m\neq m': \quad \lim_{i\to \infty} \langle \Phi^{n_i} (\psi_m), \Phi^{n_i} (\psi_{m'}) \rangle = 0 =
    \langle \mathcal{P}_\chi (\psi_m), \mathcal{P}_\chi (\psi_{m'}) \rangle,
\end{equation}
where we have used the continuity of the Hilbert-Schmidt inner product on $\M{d}$. Thus, $\{ \mathcal{P}_\chi (\psi_m)\}_{m=1}^M$ is an orthogonal set of states in $\chi (\Phi)$. Clearly, the block structure of $\chi (\Phi)$ ensures that $M\leq \sum_{k=1}^K \dim \Hil_{k,1}$. Since $\{\psi_m \}_{m=1}^M$ was an arbitrary zero-error classical code for $\Phi^N$, we obtain
\begin{equation}
    C^{(1)}_0 (\Phi^N) \leq \log \sum_{k=1}^K \dim \Hil_{k,1}.
\end{equation}
\end{proof}

\begin{remark}
    We urge the readers to check that for an ergodic channel $\Phi$, $\dim \Hil_{k,1}=1$ for all $k$ in Eq.~\eqref{eq:phasespace} \cite{Evans1978Spectral,Wolf2012Qtour}, so that $\sum_k \dim \Hil_{k,1} = \dim \chi (\Phi) =|\mathbb{T} \cap \operatorname{spec}\Phi|$ and we get 
    \begin{equation}
        \lim_{n\to \infty} C^{(1)}_0 (\Phi) = \log |\mathbb{T} \cap \operatorname{spec}\Phi|  .  \end{equation}
\end{remark}

\subsection{Zero-error quantum communication through dQMS}

In this subsection, we consider the task of storing quantum information in a system whose time evolution is governed by a dQMS $\{\Phi^n \}_{n\in \mathbb{N}}$, where $\Phi: \B{\Hil}\to \B{\Hil}$ is a quantum channel. We take a slightly different route here than what was taken in the classical case in previous subsections. We first prove the analogue of Theorem~\ref{theorem:C0phin} for the quantum capacity, from which the characterization of dQMS that eventually become useless for perfect quantum communication will follow naturally. The peripheral space $\chi (\Phi)$ will again play a crucial role in our discussion. 

\begin{theorem}\label{theorem:Q0phin}
    Let $\Phi:\B{\Hil}\to \B{\Hil}$ be a quantum channel and $\{\Phi^n\}_{n\in \mathbb{N}}$ be the associated dQMS. Then, $\exists N\leq (\dim \Hil) ^2$ such that
    \begin{equation}
        \forall q\in \mathbb{N}: \quad \lim_{n\to \infty} Q^{(1)}_0 (\Phi^n) = Q^{(1)}_0 (\Phi^N) = Q^{(1)}_0 (\Phi^{N+q}) =
        \log \max_{k} \dim \Hil_{k,1}.
    \end{equation}
\end{theorem}
\begin{proof}

As in the proof of Theorem~\ref{theorem:C0phin}, we initially note that the first two equalities follow from Theorem~\ref{stabilize}, which shows that there exists $N\leq (\dim \Hil) ^2$ such that the operator systems stabilize after time $N:$
\begin{equation}
    S_{\Phi} \subset S_{\Phi^2} \subset \ldots \subset S_{\Phi^N} = S_{\Phi^{N+1}} = \ldots = S_{\Phi^{N+q}} = \ldots
\end{equation}
    To show the final equality, first note that the action of a channel $\Phi$ on its peripheral space is reversible \cite[Theorem 6.16]{Wolf2012Qtour}, i.e., there exists a channel $\mathcal{R}:\B{\Hil}\to \B{\Hil}$ such that $\mathcal{R}\circ \Phi = \mathcal{P}_{\chi}$, where $\mathcal{P}_{\chi}$ is the channel that projects onto the peripheral space $\chi (\Phi)$. Thus, in the language of \cite{kribs2006error}, all the $\Hil_{k,1}$ sectors in the decomposition in Eq.~\eqref{eq:phasespace2} are correctable for $\Phi^n$ for all $n\in \mathbb{N}$. Corresponding subspaces $\mathcal{C}_k \subseteq \Hil$ with $\dim \mathcal{C}_k = \dim{\Hil_{k,1}}$ can then be constructed using \cite[Theorem 3.7]{kribs2006error} that are correctable for $\Phi^n$ for all $n\in \mathbb{N}$ in the sense of Defintion~\ref{def:zero-quantum}. This shows
    \begin{equation}
        \forall n\in \mathbb{N}: \quad Q^{(1)}_0 (\Phi^n) \geq 
        \log \max_{k} \dim \Hil_{k,1}.
    \end{equation}
    Conversely, suppose that a subspace $\mathcal{C}\subseteq \Hil$ is correctable for $\Phi^N$ in the sense of Definition~\ref{def:zero-quantum}. Since the operator systems of the semigroup stabilize after time $N$, the subspace $\mathcal{C}$ is also correctable for $\Phi^{N+q}$ for all $q\in \mathbb{N}$. A reformulation of the Knill Laflamme error-correction conditions in terms of the quantum relative entropy $D(\cdot \Vert \cdot)$ \cite{Petz1988sufficient, Jenov2006sufficent, Jenov2006sufficient2, Hiai2011sufficient} shows that for all states $\rho, \sigma\in \State{\Hil} \,\, \text{with}\,\, \supp (\rho) \subseteq \mathcal{C}$ and $\supp (\sigma) \subseteq \mathcal{C}$, the following is true:
    \begin{equation}
        \forall q\in\mathbb{N}: \quad D(\rho || \sigma) = D(\Phi^{N+q}(\rho)||\Phi^{N+q} (\sigma)).
        \end{equation}
    Since there exists an increasing subsequence $(n_i)_{i\in \mathbb{N}}$ such that $\Phi^{n_i}\to \mathcal{P}_{\chi}$ as $i\to \infty$, the following holds true for all states $\rho, \sigma\in \State{\Hil} \,\, \text{with}\,\, \supp (\rho) \subseteq \mathcal{C}$ and $\supp (\sigma) \subseteq \mathcal{C}$ \cite{Shirokov2022}:
    \begin{equation}
        \quad \lim_{i\to \infty} D(\Phi^{n_i}(\rho)||\Phi^{n_i} (\sigma)) = D(\rho || \sigma) = D(\mathcal{P}_{\chi}(\rho)||\mathcal{P}_{\chi} (\sigma)),
        \end{equation}
        which is equivalent to saying that the subspace $\mathcal{C}$ is correctable for $\mathcal{P}_{\chi}$ as well. This means that $\mathcal{P}_{\chi}$ acts like a $*-$homomorphism on the algebra $\B{\mathcal{C}}$ \cite[Theorem 3]{Choi2009error}, upto smearing by a fixed operator. The structure of $*-$homomorphisms between matrix algebras is well-understood \cite{davidson1996calgebra}. In particular, this means that the image $\chi(\Phi)$ of $P_{\chi}$ must be able to accommodate atleast one copy of the full matrix algebra $\B{\mathcal{C}}$, which is only possible if there exists a $k$ such that $\dim \mathcal{C}\leq \dim \Hil_{k,1}$. Since $\mathcal{C}$ is an arbitrary correctable subspace of $\Phi^N$, we get
        \begin{equation}
            Q^{(1)}_0 (\Phi^N) \leq \log \max_k \dim \Hil_{k,1}.
        \end{equation}
\end{proof}

Using the above theorem, we can easily characterize the class of \emph{eventually q-scrambling} dQMS $\{\Phi^n \}_{n\in \mathbb{N}}$, i.e., the ones for which there exists $n\in \mathbb{N}$ such that $Q^{(1)}_0 (\Phi^n)=0$.

\begin{theorem}\label{theorem:Q0phin-zero}
    Let $\Phi:\B{\Hil}\to \B{\Hil}$ be a quantum channel and $\{\Phi^n\}_{n\in \mathbb{N}}$ be the associated dQMS. Then, the following are equivalent:
    \begin{itemize}
        \item $\Phi$ is eventually \emph{q-}scrambling
        \item $\Phi$ is asymptotically entanglement-breaking.
        \item All the $\B{\Hil_{k,1}}$ blocks in the peripheral decomposition in eq.~\eqref{eq:phasespace2} are one-dimensional. 
        \item All peripheral points of $\Phi$ communte with one another: 
        $$\forall X,Y\in \chi(\Phi): \quad [X,Y]=XY - YX = 0.$$
    \end{itemize}
\end{theorem}
\begin{proof}
    From the last theorem, it is clear that $\Phi$ is eventually q-scrambling if and only if $\dim \Hil_{k,1}=1$ for all $k=1,2,\ldots ,K$ in Eq.~\eqref{eq:phasespace2}. The rest of the equivalences follow from \cite[Theorem 32]{Lami2016eb}.
\end{proof}

\subsection{Bounds on scrambling times}

We have completely classified channels $\Phi:\M{d}\to \M{d}$ (or dQMS $\{\Phi^n \}_{n\in \mathbb N}$) that eventually lose their ability to send classical or quantum information perfectly. We now consider upper bounds on the minimum time 
after which such dQMS lose their information transmission capacity. We refer to this as the scrambling time (or the scrambling index) of the dQMS.

\begin{definition}\label{def;indices}
For a channel $\Phi$ (or a  dQMS $\{\Phi^n \}_{n\in \mathbb N}$), we define
\begin{itemize}
    \item (quantum scrambling time) $q(\Phi) := \min \{n\in \mathbb{N}: Q_0^{(1)}(\Phi^n)=0\}.$
    \item (classical scrambling time) $c(\Phi) :=\min \{n\in \mathbb{N}: C_0^{(1)}(\Phi^n)=0\}.$
    \item (entanglement-assisted classical scrambling time) $c_E(\Phi) := \min \{n\in \mathbb{N}: C_{0E}^{(1)}(\Phi^n)=0\}.$
\end{itemize}
Here, we adopt the convention that the minimum of an empty set is $+\infty$.
    \end{definition}
Since for any channel $\Phi$, $Q_0^{(1)}(\Phi)\leq C_0^{(1)}(\Phi)\leq C_{0E}^{(1)}(\Phi)$, it follows that $q(\Phi)\leq c(\Phi)\leq c_E(\Phi)$. Moreover, the inequalities here can all be strict. The separation between $q(\Phi)$ and $c(\Phi)$ 
can be illustrated by considering a non-mixing channel $\Phi$ that is entanglement-breaking, so that $q(\Phi)=1$ and $c(\Phi)=c_E (\Phi)=+\infty$. Examples of this kind can be easily constructed: any classical channel $\Phi_A:\M{d}\to \M{d}$ of the form
\begin{equation}\label{eq:classicalchannel}
    \forall X\in \M{d}: \quad \Phi_A(X) = \sum_{i,j} A_{ij}X_{jj} \ketbra{i},
\end{equation}
where $A$ is an entrywise non-negative column stochastic matrix, works. Intuitively, since this channel completely decoheres its input, no quantum information can be sent through it. However, since $\mathrm{spec}\, \Phi_A = \mathrm{spec}\, A \cup \{0\},$ we can easily choose $A$ so that $\Phi_A$ is non-mixing. To show the separation between $c(\Phi)$ and $c_E (\Phi)$, we note that there exist channels $\Phi$ such that \cite{Duan2009zero}
\begin{equation}
     C_0^{(1)}(\Phi)=0 \,\,\text{ and } \,\, C_{0E}^{(1)}(\Phi)>0.
\end{equation}

We now prove the central result of this subsection.

\begin{theorem}\label{stabilize} 
    For any channel $\Phi:\M{d}\to \M{d}$, $\exists N\leq d^2 - \dim S_{\Phi}$ such that
    \begin{align}
        \forall n\in \mathbb{N}: \quad 
         Q^{(1)}_0(\Phi^N) &= Q^{(1)}_0(\Phi^{N+n}), \\ C^{(1)}_0(\Phi^N) &= C^{(1)}_0(\Phi^{N+n}), \\
         C^{(1)}_{0E}(\Phi^N) &= C^{(1)}_{0E}(\Phi^{N+n}).
    \end{align}
\end{theorem}

\begin{proof}
Let $\{K_i \}_{i=1}^p \subseteq \M{d} $ be a set of Kraus operators for $\Phi$. Recall that $S_{\Phi}=\text{span}\{K_i^{\dagger}K_j\,;\, \,1\leq i,j \leq p\}$. It can be easily checked that for any $n\in \mathbb{N}$, we have
\begin{equation}
    S_{\Phi^{n+1}} = \text{span} \{ K_i^{\dagger} X K_j : X\in S_{\Phi^n}, \, 1\le i,j \leq p\}.
\end{equation}
We thus obtain an increasing chain of operator systems $S_{\Phi} \subseteq S_{\Phi^2} \subseteq \ldots$. Furthermore, if $S_{\Phi^n}=S_{\Phi^{n+1}}$ for some $n$, then $S_{\Phi^n}=S_{\Phi^{n+k}}$ for all $k\in \mathbb{N}$. In other words, the increasing chain of operator systems stabilizes at some point. Let $N$ denote the minimum $n\in \mathbb{N}$ such that $S_{\Phi^n}=S_{\Phi^{n+1}}$. We then obtain the following chain
\begin{equation}
    S_{\Phi} \subset S_{\Phi^2} \subset \ldots \subset S_{\Phi^N} = S_{\Phi^{N+1}} = \ldots = S_{\Phi^{N+k}} = \ldots
\end{equation}
Note that the inclusions above are all strict. This is because if $S_{\Phi^n}=S_{\Phi^{n+1}}$ for some $n<N$, the above stabilization argument would contradict the minimality of $N$. Moreover, since all the operator systems are inside $\M{d}$ which is of dimension $d^2$, the maximum length of the above chain is $d^2 -\dim S_{\Phi}$. Hence, $N\leq d^2 -\dim S_{\Phi}$. Recall that all the one-shot zero-error capacities can be characterized in terms of the operator system (Theorem \ref{op-syst-char}). Since the operator system stabilizes after $N$ iterations, all the one-shot zero-error capacities also stabilize after $N$ iterations. 
\end{proof}

The above result immediately yields upper bounds on the scrambling times of all eventually scrambling evolutions. 

\begin{corollary}\label{corollary:index}
    For any mixing channel $\Phi:\M{d}\to \M{d}$, the following bound holds: 
    \begin{equation}
        q(\Phi)\leq c(\Phi)\leq c_E(\Phi) \leq d^2 - \dim S_{\Phi}.
    \end{equation}
    Furthermore, for any asymptotically entanglement-breaking channel $\Phi:\M{d}\to \M{d}$, we have $q(\Phi)\leq d^2$.
\end{corollary}
\begin{proof}
    We know that for a mixing channel $\Phi$, $\exists k\in \mathbb{N}$ such that $C^{(1)}_{0E}(\Phi^k)=0$. Hence, $C^{(1)}_{0E}(\Phi^{k+n})=0$ for any $n\in \mathbb{N}$. By the above proposition, we must have $k \leq d^2 - \dim S_\Phi$, which implies that 
    $c_E(\Phi) \leq d^2 - \dim S_{\Phi}$.
    %Hence the above chain stabilises to $\M{d}$ after this $k$-iterations and hence by the above proposition $c_E(\Phi) \leq d^2 - \dim S_{\Phi}$. 
    Since $q(\Phi)\leq c(\Phi)\leq c_E(\Phi)$ for any channel $\Phi$, we have the required chain of inequalities. The result for asymptotically entanglement-breaking channels follow similarly.
\end{proof}

Let us take a moment to note that the dimension factor of $d^2$ in Proposition~\ref{stabilize} and Corollary~\ref{corollary:index} is optimal for the classical scrambling times. To show this, consider the $d\times d$ stochastic matrix  

\begin{equation}\label{eq:Ad}
    A_d := \left( \begin{array}{cccccc}
        0 & 1/2 & 0 & 0 & 0 & 0 \\
        0 & 0 & 1 & 0 & 0 & 0 \\
        0 & 0 & 0 & 1 & 0 & 0 \\
        \vdots & \vdots & \vdots & \vdots & \ddots & \vdots \\
        0 & 0 & 0 & 0 & 0 & 1 \\
        1 & 1/2 & 0 & 0 & 0 & 0 \\
    \end{array} \right).
\end{equation}
Firstly, observe that for a classical channel $\Phi_{A}$ of the form defined in Eq.~\eqref{eq:classicalchannel}, the two scrambling times $c(\Phi_A)$ and $c_E(\Phi_A)$ are equal.

\begin{lemma}\label{lemma:c=cE}
    For a classical channel $\Phi_A:\M{d}\to \M{d}$, $c(\Phi_A)=c_E(\Phi_A)$
\end{lemma}
\begin{proof}
    It suffices to show that $C^{(1)}_0(\Phi_A)=0\implies C^{(1)}_{0E}(\Phi_A)=0$. Hence, assume that $C^{(1)}_0(\Phi_A)=0$. This means that the operator system $S_{\Phi_A}$ is such that there are no rank one matrices in $S_{\Phi_A}^{\perp}$. However, since the operator system is `graphical' \cite{Duan2013noncomm}, i.e.,
    \begin{equation}
        S_{\Phi_A} = \text{span} \{ \ketbra{i}{j} : i=j \text{ or } i \sim_Aj \},
    \end{equation}
    where the notation $i\sim_A j$ is used to denote that $i$ and $j$ are confusable under the transition probabilities defined by $A$, the absence of a rank one element in $S_{\Phi_A}^{\perp}$ implies that $S_{\Phi_A}^{\perp}=\{0\}$, which shows that $C^{(1)}_{0E}(\Phi_A)=0$. 
\end{proof}

The results in \cite{Akelbek2009index} then show that 
\begin{equation}\label{eq:A_d}
    c(\Phi_{A_d})=c_E(\Phi_{A_d}) = \ceil[\bigg]{\frac{d^2 -2d +2}{2}}.
\end{equation}

\begin{remark}
    The optimal upper bound on the scrambling times of classical stochastic matrices $A\in \M{d}$ is already known in the literature \cite{Akelbek2009index, Guterman2019index}. It is of the form noted above:
    \begin{equation} 
    c(A) \leq \ceil[\bigg]{\frac{d^2 -2d +2}{2}},
    \end{equation}
where $\ceil{\cdot}$ denotes the ceiling function and equality is attained for $A=A_d$. However, this result is derived in a purely combinatorial framework, with no reference made to any zero-error information transmission task. Moreover, the proof is long and uses a variety of intricate graph-theoretic techniques. In contrast, the proof of Proposition~\ref{stabilize} proceeds via a simple operator theoretic chain argument, yields an upper bound with the same optimal $d^2$ dimension factor, and works not only for classical channels but also for quantum channels.
\end{remark}

%\begin{remark}
%    The above result improves the bound for the scrambling index of a channel given in \cite[Theorem 7.7]{Dobrushin}. Indeed, they showed by a dimension-counting argument that $c(\Phi)\leq d^2-1$. Since in general 
%    $dim (S_\Phi)>1$, our bound improves the result obtained in \cite{Dobrushin}.
%\end{remark}

\section{Auxilliary results}
\subsection{Ergodicity and invariant subspaces}

In this section, we study equivalent descriptions of ergodic and mixing quantum channels in terms of their invariant subspaces.

\begin{definition} A subspace $S \subseteq {\mathbb{C}}^d$ is said to be invariant under $\Phi :\M{d} \to \M{d}$ if for all states $\rho$ with $\supp \, (\rho) \subseteq {S}$, $\supp\, \Phi(\rho) \subseteq S$. A subspace $S \subseteq {\mathbb{C}}^d$ is a minimal invariant subspace of $\Phi$ if for any subspace ${S}^\prime \subseteq {\mathbb{C}}^d$ which is invariant under $\Phi$, $S \subseteq {S}^\prime$.
\end{definition}

The following characterization of ergodicity was obtained in  \cite[Theorem 1]{Burgarth2013ergodic}.

\begin{theorem}\label{theorem:ergodic-inv}
A channel $\Phi :\M{d} \to \M{d}$ is ergodic if and only if $\Phi$ admits a non-zero minimal invariant subspace $\Ss$. Moreover, $\Ss$ is precisely the support of the unique invariant state of $\Phi$.
\end{theorem}
%\begin{proof}
%\begin{itemize}\noindent\item 
%{\nilan{Suppress proof}}

%($ \Rightarrow$) Assume $\Phi$ is ergodic. Then there exists a unique state $\rhos$ satisfying $\Phi(\rhos) = \rhos$. Then $S = \supp \, \rhos$ is invariant, since for any state $\sigma$ with $\supp\,\sigma \subseteq {S}$, $\exists \epsilon > 0$ such that 
%$$ \epsilon \sigma \leq \rho \implies  \epsilon \Phi(\sigma) \leq \Phi(\rho) \implies \supp \Phi(\sigma) \subseteq S.$$
%Moreover, $S$ is minimal since for any other invariant subspace ${S}^\prime$, the restriction of $\Phi$ to this subspace (denoted as $\Phi\big|_{{S}^\prime}$) admits a fixed point $\rho^\prime$ with $\supp\,\rho^\prime \subseteq {S}^\prime$. However, since $\Phi$ is ergodic, we must have $\rho^\prime = c \rho$. This in turn implies that
%$$S = \supp\, \rho = \supp \, \rho^\prime \subseteq {S}^\prime.$$

%($\Leftarrow$) Assume that $\Phi$ admits a non-zero minimal invariant subspace ${S}\subseteq {\mathbb{C}}^d$ but is not ergodic. Then there exists two distinct fixed points, say $\rho_1$ and $\rho_2$. Then $( \rho_1 - \rho_2)_+$ and $( \rho_1 - \rho_2)_-$ are two fixed points of mutually orthogonal supports. This implies that  $S_+:= \supp\, (\rho_1 - \rho_2)_+$ and $S_-:= \supp\, (\rho_1 - \rho_2)_-$ are two invariant subspaces and $S_+$ is orthogonal to $S_-$. Since $S$ is minimal, $S \subseteq S_+$ and $S \subseteq S_-$, and hence $S \subseteq S_+ \cup S_- = \{0\}$. This in turn implies that $S=\{0\}$, which leads to a contradiction.
    %\end{itemize}
 %   \end{proof}

In what follows, we provide a similar characterization for the class of mixing quantum channels.

\begin{theorem}\label{thm:inv-sub}
For a quantum channel $\Phi:\M{d}\to \M{d}$, the following are equivalent.
\begin{enumerate}
\item
$\Phi$ is mixing.
\item $\Phi$ admits a non-zero minimal invariant subspace $\Ss \subseteq {\mathbb{C}}^d$ and $\exists N \in {\mathbb{N}}$ such that $\forall \rho\in \St{d}$ and $n \geq N$, $\supp \, \ps\Phi^n (\rho) \ps = \Ss$, where $\ps$ denotes the orthogonal projector onto $\Ss$.
\item $\Phi$ admits a non-zero minimal invariant subspace $\Ss\subseteq {\mathbb{C}}^d$ and $\exists N \in {\mathbb{N}}$ such that for $n \geq N$, $\ps {\mathcal{K}}_n (\Phi) = \ps \M{d}$, where ${\mathcal{K}}_n (\Phi) := {\rm{span}}\, \{K_{i_1} \ldots K_{i_n}\}$. Here, $\{K_i\}_i$ denotes a set of Kraus operators of $\Phi$ and $\ps$ is the orthogonal projector onto $\Ss$.
\end{enumerate}
\end{theorem}

\begin{remark}
Note that if $\ps \mathcal{K}_n(\Phi) = \ps \M{d}$, then $\ps {\mathcal{K}}_m(\phi) = \ps \M{d}$ $\forall m \geq n$, since 
$$ \forall X \in \M{d}: \,\, \ps X = \sum_{i_i \ldots i_n} c_{i_1\ldots i_n} \ps K_{i_1}\ldots K_{i_n},$$ and for each term in the sum, $\ps K_{i_1}\ldots K_{i_{n-1}} \in \ps \M{d},$ and so can also be written as 
$$ \ps K_{i_1}\ldots K_{i_{n-1}} = \sum_{j_1\ldots j_n} d_{j_1 \ldots j_n} \ps K_{j_1}\ldots K_{j_n}.$$
\end{remark}

\begin{proof}
%\begin{itemize}
%\item
   $(1) \implies (3)$: Assume $\Phi$ is mixing. Then $\Phi$ is ergodic and hence admits a non-zero minimal invariant subspace $\Ss = \supp (\rhos)$, where $\rhos\in \St{d}$ is the unique invariant state of $\Phi$ (Theorem~\ref{theorem:ergodic-inv}). On the contrary, assume that $\forall n \in {\mathbb{N}}$, $\ps {\mathcal{K}}_n (\Phi) \subset \ps \M{d}$ with the containment being strict. Let us choose an operator $X_n \in \M{d}$ such that
   \begin{align}
       0 \neq \ps X_n \in \left( \ps {\mathcal{K}}_n (\Phi) \right)^\perp,
   \end{align}
   where the orthogonal complement is taken inside $\ps \M{d}$, so that $\forall\ps K^{(n)} \in \ps {\mathcal{K}}_n (\Phi) $, 
   \begin{align}
   \langle \ps X_n, \ps K^{(n)} \rangle = \Tr \left( X_n^* \ps \ps K^{(n)}\right) = 0.
   \label{ortho}
   \end{align}
   Note that $\ps X_n X_n^* \ps$ is a non-zero positive operator supported in $\Ss$. Hence,
   \begin{align}\label{eq:mixing}
        \forall n\in \mathbb{N}: \quad \Tr \left(\rhos \ps X_n X_n^* \ps\right) \geq \frac{1}{||\rhos^{-1}||_\infty}  \Tr \left(\ps X_n X_n^* \ps\right).
    \end{align}
However, we can also write
\begin{align}\label{eq:mixing1}
  \bigg| \Tr (\rhos \ps X_n X_n^* \ps) \bigg| &= \bigg|\sum_{i_1,\ldots, i_n} |\Tr(X_n^* \ps \ps K_{i_1}\ldots K_{i_n})|^2 - \Tr (\rhos \ps X_n X_n^* \ps)\bigg|.
\end{align}
Let us consider the two terms on the RHS of \eqref{eq:mixing1} separately. Using cyclicity of the trace 
and the following elementary relations for $A \in \M{d}$:
\begin{align}
\Tr A &= \Tr [(A \otimes I) \Omega] ,\label{eq:mixing2} \\
\Tr A\Tr A^* &= \Tr \left[\Omega (A \otimes I)\Omega (A^* \otimes I)\right], \label{eq:mixing3} \\
(A \otimes I) \ket{\Omega} &= (I \otimes A^T) \ket{\Omega} \label{eq:mixing4}, 
\end{align}
where $\ket{\Omega} = \sum_{i=1}^d \ket{ii}$ is the unnormalized maximally entangled state and $\Omega = \ketbra{\Omega}$, we can write
\begin{align}
  & \sum_{i_1,\ldots, i_n} |\Tr(X_n^* \ps \ps K_{i_1}\ldots K_{i_n})|^2 \\
& =
\sum_{i_1,\ldots, i_n} \Tr( \ps K_{i_1}\ldots K_{i_n}X_n^* \ps) \Tr( \ps X_n (K_{i_1}\ldots K_{i_n})^* \ps) \\
&= \sum_{i_1,\ldots, i_n}  \Tr \left[\Omega (\ps K_{i_1} \ldots K_{i_n} X_n^* \ps \otimes \iden) \Omega (\ps X_n (K_{i_1} \ldots K_{i_n})^* \ps \otimes \iden) \right],\\
&= \Tr \left[\Omega (\tPhin \otimes \id) (X_n^* \ps \otimes \iden) \Omega( \ps X_n \otimes \iden)\right]. \label{eq:mixing5}
\end{align}
Here, $\tPhin: \M{d} \to \M{d}$ is the CP map defined as follows:
\begin{align}
    \tPhin(X) = \ps \Phi_n(X) \ps = \ps\left( \sum_{i_1,\ldots, i_n} (K_{i_1}\ldots K_{i_n})X (K_{i_1}\ldots K_{i_n})^*\right) \ps.
\end{align}
To deal with the second term in \eqref{eq:mixing1} we further define a completely depolarizing map $\Phi_\infty: \M{d} \to \M{d}$ through the relation $\Phi_\infty(X) = \Tr (X) \rhos$, so that $\tPhin(X) \longrightarrow \Phi_\infty(X)$ as $n \to \infty$, since
\begin{align}
    \forall X \in \M{d}: \quad \tPhin(X) = \ps \Phi^n(X) \ps \longrightarrow \ps (\Tr X) \rhos \ps = (\Tr X) \rhos \quad {\hbox{as $n \to \infty$}}.  
    \end{align}
Then, once again using the relations \eqref{eq:mixing2},\eqref{eq:mixing3},\eqref{eq:mixing4}, it can be shown that
\begin{align}
        \Tr (\rhos \ps X_n X_n^* \ps) = \Tr \left(\Omega (\Phi_\infty \otimes \id) (X_n^* \ps \otimes I) \Omega( \ps X_n \otimes I) \right). \label{eq:mixing6}
    \end{align}

\begin{comment}
   and noting that $\Phi_\infty(\ket{i}\bra{j})= \delta_{ij} \rhos$, the second term on the RHS of (\ref{eq:mixing1}) can be expressed as follows
    \begin{align}
        \Tr (\rhos \ps X_n X_n^* \ps) &= \Tr \left[ \Omega (\rhos \otimes (X_n^* \ps)^T (\ps X_n)^T) \right] \\
        &= \Tr \left(\Omega \sum_{i,j} \delta_{ij} \rhos \otimes  (X_n^* \ps)^T \ket{i}\bra{j}(\ps X_n)^T  \right)\\
         &= \Tr \left(\Omega (\Phi_\infty \otimes \id) (I \otimes (X_n^* \ps)^T ) \Omega( \ps I \otimes (\ps X_n)^T I) \right) \\
        &= \Tr \left(\Omega (\Phi_\infty \otimes \id) (X_n^* \ps \otimes I) \Omega( \ps X_n \otimes I) \right) \label{eq:mixing6}
    \end{align} 
\end{comment}
Thus, we can express Eq.~\eqref{eq:mixing1} as follows:
    \begin{align}
          \bigg| \Tr (\rhos \ps X_n X_n^* \ps) \bigg| & = \bigg| \Tr \left[ \Omega ( (\tPhin - \Phi_\infty) \otimes \id) (X_n^* \ps \otimes I) \Omega( \ps X_n \otimes I) \right]  \bigg| \\
          & \le \| \Omega\|_\infty \|(\tPhin - \Phi_\infty) \otimes \id \|_1 \Tr (\ps X_n X_n^* \ps).
    \end{align}
    Note that, $\|(\tPhin - \Phi_\infty) \otimes \id \|_1 \to 0$ as $n \to \infty$ since $\tPhin \mapsto \Phi_\infty$ in this limit, 
    which contradicts Eq.~\eqref{eq:mixing}. This concludes the proof of $(1) \implies (3)$.
\medskip

\noindent
    $(3) \implies (2)$: Assume that $\ps \mathcal{K}_n (\Phi) = \ps \M{d}$ for some $n \in {\mathbb{N}}$. Hence, $\forall \ket{\psi}$,
    $$ \ps\Phi^n\left( \ketbra{\psi}\right) \ps\quad = \sum_{i_1, \ldots, i_n} \ps K_{i_1} \ldots K_{i_n} \ketbra{\psi } (K_{i_1} \ldots K_{i_n} )^* \ps.$$
    This implies that
    \begin{align}
        \supp \, \ps\Phi^n \left( \ketbra{\psi}\right) \ps & = {\rm{span}} \{ \ps K_{i_1} \ldots K_{i_n} \ket{\psi}\}\nonumber\\
        &= \ps{\rm{span}} \{  K_{i_1} \ldots K_{i_n}\} \ket{\psi}\nonumber\\
        &= \ps {\mathcal{K}}_n(\Phi) \ket{\psi} = \ps \M{d} \ket{\psi} = \Ss.
    \end{align}
    \smallskip

    \noindent
    $(2) \implies (1)$: Assume that $\Phi$ admits a minimal invariant subspace $\Ss \neq \{0\}$ and $\exists N \in {\mathbb{N}}$ such that for any $\rho\in\St{d}$ and $n \geq N$, $\supp\, \ps \Phi^n (\rho) \ps = \Ss$. Since $\Phi$ admits a minimal invariant subspace, it must be ergodic. Assume on the contrary that $\Phi$ is not mixing. Then, there must exist a peripheral eigenvalue different from $1$.  Hence, 
    $$\mathbb{T} \cap {\rm{spec}} \,\Phi = \left\{ e^{\frac{2\pi i m}{q}} \, : \, m=0,1,\ldots, q-1\right\} $$ 
    for some $q \in {\mathbb{N}}$, $q \geq 2$ (see Lemma~\ref{lemma:ergodic-spectrum}). This implies that $\Phi^q$ has $q$ distinct fixed points. Let us consider two of them, say $\rhos, \sigma\geq 0$, where $\supp \rhos = \Ss$. If $\supp \ps \sigma \ps \neq \Ss$, we can choose $k$ large enough such that $kq\geq N$ and $\supp \ps \Phi^{kq}(\sigma) \ps = \supp \ps \sigma \ps \neq \Ss$, leading to a contradiction. So, we can assume that $\supp \ps \sigma \ps = \Ss$. Now, consider $\omega = \rhos - \epsilon \sigma$, which is again a fixed point of $\Phi^q$. Clearly,
    \begin{align}
        \ps \omega \ps \geq  0 &\iff \rhos \geq \epsilon \ps \sigma \ps \\
        &\iff \ps \geq \epsilon \rhos^{-1/2}\sigma \rhos^{-1/2} \\
        &\iff \epsilon \leq \frac{1}{\norm{\rhos^{-1/2}\sigma \rhos^{-1/2}}_{\infty}}.
    \end{align}
Let us choose $\epsilon=1/\norm{\rhos^{-1/2}\sigma \rhos^{-1/2}}_{\infty}$ and $0\neq \ket{v}\in \Ss$ such that $\rhos^{-1/2}\sigma \rhos^{-1/2} \ket{v} = \ket{v}/\epsilon.$ Then, 
    \begin{align}
      \epsilon \ps \sigma \rhos^{-1/2}\ket{v} =  \rhos^{1/2}\ket{v} = \rhos \rhos^{-1/2}\ket{v} \implies (\rhos - \epsilon\ps \sigma \ps) \rhos^{-1/2}\ket{v}=0.
    \end{align}
This means that $\ps \omega \ps = \rhos - \epsilon \ps \sigma \ps$ has a kernel in $\Ss$, so $\supp \ps \omega \ps \neq \Ss$. Since $\omega$ is fixed by $\Phi^q$, its positive and negative parts $\omega_\pm \geq 0$ are also fixed by $\Phi^q$. Finally, by choosing $k$ large enough such that $kq\geq N$, we get 
$\supp \ps \Phi^{kq}(\omega_\pm) \ps = \supp \ps \omega_\pm \ps \subseteq \supp \ps \omega \ps \neq \Ss$, which contradicts our original assumption. Hence, $\Phi$ must be mixing, and the proof is complete.
    %\end{itemize}
    \end{proof}

It is easy to see that if the quantum channel is primitive, the equivalences obtained in Theorem~\ref{thm:inv-sub} reduce to those given in \cite[Proposition 3]{wielandt2010}, which we restate below.

\begin{corollary}\label{corollary:primitive}
    For a quantum channel $\Phi:\M{d}\to \M{d}$, the following are equivalent.
\begin{enumerate}
\item
$\Phi$ is primitive.
\item $\exists N \in {\mathbb{N}}$ such that $\forall \rho\in \St{d}$ and $n \geq N$, $\Phi^n  (\rho)$ is of full rank. 
\item $\exists N \in {\mathbb{N}}$ such that for $n \geq N$, ${\mathcal{K}}_n (\Phi) =  \M{d}$, where ${\mathcal{K}}_n (\Phi) := {\rm{span}}\, \{K_{i_1} \ldots K_{i_n}\}$. Here, $\{K_i\}_i$ denotes a set of Kraus operators of $\Phi$.
\end{enumerate}
\end{corollary}

\subsection{Connection with the quantum Wielandt index}

Consider a primitive channel $\Phi:\M{d}\to \M{d}$. Then, Corollary~\ref{corollary:primitive} informs us that $\exists N \in \mathbb{N}$ such that $\Phi^n$ is strictly positive for $n\geq N$. This leads naturally to the definition of the \emph{Wielandt index} \cite{wielandt2010, rahaman2020, Wiel-revisit, jia-capel} of $\Phi$:
\begin{equation}
    w(\Phi) := \min \{n\in \mathbb{N}: \Phi^n \,\, \text{is strictly positive} \}.
\end{equation} 

Note that if a channel $\Phi:\M{d}\to \M{d}$ is strictly positive, then $\Tr[\Phi(\rho)\Phi(\sigma)]>0$ for all $\rho,\sigma\in \St{d}$. Thus, $\Phi$ is also scrambling. Hence for a primitive channel $\Phi$, it holds that  $$q(\Phi)\leq c(\Phi)\leq w(\Phi).$$ 

In this section, we obtain a converse bound $c(\Phi)\leq w(\Phi)\leq (d-1)c(\Phi)$ for primitive channels $\Phi:\M{d}\to \M{d}$ that are also unital, thus establishing a close link between the scrambling times and Wielandt indices for such channels. We do this by relating the notions of scrambling and strictly positivity. We show that for unital channels that are scrambling, there is a linear universal bound (depending only on the system dimension) for the channel iterations to become strictly positive. In order to prove this result, we first need some new definitions.

Two projections $P,Q \in \M{d}$ are said to be (Murray-von Neumann) equivalent (denoted $P \sim Q$) if there is an operator $V \in \M{d}$ such that $P= VV^*$ and $Q= V^* V$. Hence, $P \sim Q$ if and only if $\Tr P = \Tr Q$. Further, we say that a projection $P \in \M{d}$ is {\em{non-trivial}} if $P \not\in \{0, \iden\}$. Let $P^\perp:= I - P$. The following definition is from \cite{rahaman2020} (see also \cite{idel}). 
\begin{definition}
    A quantum channel  $\Phi:\M{d}\to \M{d}$ is said to be {\em{fully irreducible}}\footnote{Such channels are also called {\em{fully indecomposable}}~\cite{idel}.} if there does not exist any pair of non-trivial, equivalent projections, $P \sim Q$, such that $\Phi(P) \leq \lambda Q$, for some $\lambda>0$.
\end{definition}

\begin{proposition}
    If a unital channel $\Phi: \M{d} \to \M{d}$ is scrambling, then it is fully irreducible.
\end{proposition}
\begin{proof}
    We prove that if $\Phi$
    is not fully irreducible, then it cannot be scrambling. To do this, let us assume that there exists a pair of non-trivial, equivalent projections $P\sim Q$ such that $\Phi(P) \leq \lambda Q$ for some $\lambda >0$. This implies that $Q^\perp \Phi(P) Q^\perp = 0.$ Let $\{K_i\}_i$ denote a set of Kraus operators of $\Phi$. Then, since $P$ is a projection (i.e.~$P=P^*$ and $P^2 = P$), we get
    \begin{align}
        Q^\perp \left( \sum_i K_i P K_i^*\right) Q^\perp &= 0 \nonumber\\
        i.e. \quad  \sum_i (Q^\perp K_i P) (Q^\perp K_i P)^* &= 0
        \nonumber\\
         \implies \quad  Q^\perp K_i P &= 0 \quad \forall \, i\nonumber\\
           \implies \quad  K_i P &= Q K_i P\quad \forall \, i 
           \label{eqq1}
    \end{align}
    Hence,
    \begin{align}
        \Phi(P) = \sum_i K_i P K_i^* &= \sum_i K_i P (K_iP)^* = Q\left(\sum_i K_i P K_i^*\right) Q = Q \Phi(P) Q \leq Q,
    \end{align}
    where the third equality follows from~\eqref{eqq1}, and the last inequality follows from the fact that $P \leq \iden$ and since $\Phi$ is unital, $\Phi(P) \leq \Phi(\iden) = \iden$. Thus we have established the following:
    \begin{align}
        \Phi(P) \leq \lambda Q \, \implies \, \Phi(P) \leq Q.
        \label{eqq2}
    \end{align}
    From \eqref{eqq2} we have $\Tr \Phi(P) = \Tr P \leq \Tr Q$, since $\Phi$ is trace-preserving. However, by assumption, $\Tr P = \Tr Q$ (since $P \sim Q$). Hence, by faithfulness of the trace we get $\Phi(P) = Q$. However, this equality implies that $\Phi$ violates the property of scrambling, as is shown explicitly below. 

    Set $\rho = P/\Tr P$ and $\sigma = \iden /d = \frac{P + P^\perp}{d}$. Then,
    \begin{align}
        ||\rho- \sigma||_1 & = || \frac{P}{\Tr P} - \frac{P}{d} -  \frac{P^{\perp}}{d} ||_1 \nonumber\\
        & = \alpha ||P||_1 + \frac{1}{d} ||P^\perp||_1, \nonumber\\
        &= \alpha \Tr P + \frac{1}{d} \Tr P^\perp,
        \label{eq:lhs}
    \end{align}
    where $\alpha := | \frac{1}{\Tr P} - \frac{1}{d}|$. On the other hand, since $\Phi$ is unital and $\Phi(P) = Q$, we get
    \begin{align}
        ||\Phi(\rho) - \Phi(\sigma)||_1 &= ||\frac{\Phi(P)}{\Tr P} - \frac{\iden}{d}||_1 \nonumber\\
        &= ||\frac{Q}{\Tr P} - \frac{Q}{d} - \frac{Q^\perp}{d}||_1\nonumber\\
        & = \alpha \Tr Q + \frac{1}{d} \Tr Q^\perp\nonumber\\
        &= \alpha \Tr P + \frac{1}{d} \Tr P^\perp,\label{eq:rhs}
    \end{align}
    where we used the fact that $\Tr P^\perp = \Tr Q^\perp$ since $\Tr P= \Tr Q$.
 Hence, from \eqref{eq:lhs} and \eqref{eq:rhs} we have
 $$||\Phi(\rho) - \Phi(\sigma)||_1 = ||\rho - \sigma||_1,$$
 and hence the quantum channel $\Phi$ is not scrambling. This concludes the proof.
\end{proof}
One might wonder whether the converse of the above proposition is true. The following example shows that this is not the case.
\begin{example}
    Let $\Phi: \M{4}\rightarrow \M{4}$ be a unital quantum channel defined as follows:
    \[ \forall X\in \M{4}: \quad \Phi(X)=1/2 \begin{bmatrix}
    X_{11}+X_{22} & 0 & 0 & 0\\
    0 & X_{22}+X_{33} & 0 & 0\\
    0 & 0 & X_{33}+X_{44} & 0\\
    0 & 0 & 0 & X_{44}+X_{11}
    \end{bmatrix}.
    \]
    It is easy to see that $\Phi$ is fully irreducible, since it sends any projection to a positive semi-definite matrix of rank strictly larger than the rank of the input projection. However, $$\Phi(\ketbra{1})=\frac{1}{2}(\ketbra{1}+\ketbra{4}) \quad \text{and} \quad \Phi(\ketbra{3})= \frac{1}{2} (\ketbra{2}+ \ketbra{3}).$$  Hence, $\Tr(\Phi(\ketbra{1})\Phi(\ketbra{3}))=0$ and it follows that $\Phi$ is not scrambling (see Theorem \ref{thm:scrambling}).
\end{example}

The following corollary provides the main result of the subsection.

\begin{corollary}\label{strict-pos}
Let $\Phi: \M{d}\rightarrow \M{d}$ be a unital quantum channel and suppose that $\Phi$ is scrambling. Then $\Phi^{d-1}$ is strictly positive. Consequently, for any primitive unital channel $\Phi$, we get $c(\Phi)\leq w(\Phi)\leq (d-1)c(\Phi)$.
\end{corollary}
\begin{proof}
It is known that a fully irreducible unital channel is strictly rank increasing, i.e., for any singular positive semi-definite $A\in \M{d}: \mathrm{Rank}(\Phi(A))> \mathrm{Rank}(A)$ (see \cite[Theorem 3.7]{rahaman2020}).

Now, since $\Phi$ is scrambling, the previous result shows that $\Phi$ is fully irreducible, and hence also strictly rank increasing. Thus, starting from any rank one projection $P$, it requires at most $d-1$ iterations of $\Phi$ to send $P$ to an invertible matrix. This concludes the proof. 
\end{proof}
\smallskip

\noindent
{\bf{Remark:}} The above corollary provides an upper bound on the number of iterations for a scrambling channel that are needed to guarantee that it becomes strictly positive. This upper bound depends solely on the dimension $d$ of the underlying space. It is natural to ask whether this bound can be improved. Note that from the definition of scrambling it follows that if $\Phi$ is scrambling, then $\Phi^*\circ\Phi$ is strictly positive. So for self-adjoint scrambling channels, it holds that $\Phi^2$ is strictly positive. This raises the following question: Could it be true that for any scrambling channel $\Phi$, $\Phi^2$ is strictly positive? However, this is not true, as shown below by an explicit example of a classical channel, defined by the column stochastic matrix, $A$:
    \begin{equation}
        A = \left(\begin{array}{ccccc}
         1/3 & 1/3 & 0 & 0 & 1/3  \\
         0 & 1/3 & 1/3 & 0 & 1/3 \\
         1/3 & 0 & 1/3 & 1/3 & 0  \\
         0 & 1/3 & 0 & 1/3 & 1/3 \\
         1/3 & 0 & 1/3 & 0 & 1/3
    \end{array}  \right).
    \end{equation}
    $A$ is clearly scrambling but $A^2$ is not strictly positive.

\subsection{Linear bounds on indices for channels with extra structure}

If the mixing channel $\Phi$ is unital and its operator system $S_{\Phi}$ is such that $S_{\Phi^n}$ is an algebra for all $n\in \mathbb{N}$, we can provide better bounds on its scrambling indices.

%It is not hard to see that if the channel $\Phi$ is strictly positive (meaning sends positive semidefinite elements of $M_d$ to positive definite ones), then $Q_0^{(1)}(\Phi)=0$. This can be seen utilizing the Knill-Laflamme condition (\cite{knil-laf}) for error-correction (see Theorem 5 in \cite{rahaman2020}). This is also true for the zero-error classical capacity (see \cite{wielandt2010}).
%Hence, for a \textbf{primitive channel} $\Phi$, there exists a finite $n\in \mathbb{N}$ such that $C_0^{(1)}(\Phi^n)=0$ and a finite $k\in \mathbb{N}$ such that $Q_0^{(1)}(\Phi^k)=0$.

In order to state and prove our main result, we need the following lemma.

\begin{lemma}\label{lemma:chain suubalgebra}
    Let $\Phi: \M{d}\rightarrow \M{d}$ be a unital channel, and let $\mathcal{M}_{\Phi^k}$ denote %the corresponding 
    the multiplicative domain of $\Phi^k$ for each $k\in\mathbb N$, and let  $\mathcal{M}_{\Phi^{\infty}}=\bigcap_{k\geq 1} \mathcal{M}_{\Phi^k}$ be the stabilized multiplicative domain as introduced in Section~\ref{subsec-fixed-pt-mult.}. 
    \begin{enumerate}

\item The stabilized multiplicative domain can be described as follows
\[\mathcal{M}_{\Phi^{\infty}}=\mathrm{alg}\{A\in \M{d}: \Phi(A)=\lambda A; |\lambda|=1\}.\]
\item     It holds that for all $k\geq 1$, $\mathcal{M}_{\Phi^k}=[S_{\Phi^k}]'$. 

\item If $\Phi$ is primitive, then $\mathcal{M}_{\Phi^{\infty}}=\mathbb{C}\iden $, i.e., the trivial algebra with only scalars. Furthermore, the containments in the following chain of subalgebras are proper: 
 \[ \mathcal{M}_{\Phi}\supsetneq \mathcal{M}_{\Phi^2} \supsetneq \cdots \mathcal{M}_{\Phi^n}\supsetneq \cdots \supsetneq \mathbb{C}\iden.\]
\end{enumerate}
\end{lemma}
\begin{proof}
    The proofs of the first statement above can be found in \cite{rahaman2017}. To prove the second assertion, recall from Section~\ref{subsec-fixed-pt-mult.} that for a unital channel $\Phi$,
\[\mathcal{M}_{\Phi}=\mathrm{Fix}_{(\Phi^*\circ\Phi)}=S_{\Phi}',\]
and this relation holds for every $\Phi^k$, $k\in\mathbb N$. 
    %follows 
    %from the fact that for any unital channel $\Phi$, $\mathrm{Fix}_{\Phi}=S_{\Phi}'$
    %simply from the {\nilan{following relation between the multiplicative domain and the fixed point algebra}}
    %\[\mathcal{M}_{\Phi^k}=\mathrm{Fix}_{(\Phi^{*k}\circ\Phi^k)}.\]

Here we prove the last statement. If $\Phi$ is unital and primitive, then $\iden$ is its only peripheral eigenoperator, and hence $(1)$ shows that $\mathcal{M}_{\Phi^{\infty}}=\mathbb{C}\iden$. In fact, having $\mathcal{M}_{\Phi^{\infty}}=\mathbb{C}\iden$ can be an alternative characterization of primitivity for unital channels (see Corollary 3.5 in \cite{rahaman2017}). 

To prove the containment of the subalgebras is proper we analyze the behaviour of multiplicative domain under composition of two channels. From \cite[Lemma 2.3]{rahaman2017}, it holds that
\[\mathcal{M}_{\Phi^k}=\{a\in \mathcal{M}_{\Phi^{(k-1)}}| \Phi^{(k-1)}(a)\in \mathcal{M}_{\Phi}\}=\{a\in \mathcal{M}_{\Phi}| \Phi(a)\in \mathcal{M}_{\Phi^{(k-1)}}\}.\]
Hence, $\mathcal{M}_{\Phi^{(k+1)}}\subseteq \mathcal{M}_{\Phi^k}$, for all $k\in \mathbb{N}$ and if $x\in \mathcal{M}_{\Phi^k}$, then $\Phi(x)\in \mathcal{M}_{\Phi^{(k-1)}}$. 

Now suppose $\mathcal{M}_{\Phi^k}= \mathcal{M}_{\Phi^{(k-1)}}$, for some $k\geq 2$. From the above observation, we know that $x\in \mathcal{M}_{\Phi^k}\implies \Phi(x)\in \mathcal{M}_{\Phi^{(k-1)}}=\mathcal{M}_{\Phi^k}$.
Since $\mathcal{M}_{\Phi^k}\subseteq \mathcal{M}_{\Phi}$, it follows that $x\in \mathcal{M}_{\Phi^k} \implies x\in \mathcal{M}_{\Phi}$ and $\Phi(x)\in \mathcal{M}_{\Phi^k} \implies x \in \mathcal{M}_{\Phi^{(k+1)}}$.
So $\mathcal{M}_{\Phi^k}= \mathcal{M}_{\Phi^{(k+1)}}$. Since by primitivity we know $\mathcal{M}_{\Phi^n}=\mathbb{C}\iden$, for large $n$, it must be the case that $\mathcal{M}_{\Phi^{(k+1)}}\subset \mathcal{M}_{\Phi^k}$, unless the latter set is just the trivial algebra.
\end{proof}

Let us also note another lemma, which is obtained by combining Propositions~\ref{prop:C0E=0} and \ref{prop:C0Q0=0}
\begin{lemma}\label{lemma:C0C0E=0}
    Let $\Phi:\M{d}\to \M{d}$ be a channel such that $S_{\Phi}$ is an algebra. Then, 
    $$C^{(1)}_{0E}(\Phi)=0 \iff C^{(1)}_0(\Phi)=0\iff [S_\Phi]'=\mathbb{C}\iden .$$
\end{lemma}
\begin{proof}
    The second equivalence is contained in Proposition~\ref{prop:C0Q0=0}. The $(\implies)$ implication in the first equivalence is trivial to show. Thus, it suffices to prove that $[S_\Phi]'=\mathbb{C}\iden \implies C^{(1)}_{0E}(\Phi)=0$. So, assume that $[S_\Phi]'=\mathbb{C}\iden$ and note that since $S_{\Phi}$ is an algebra, the double commutant theorem shows that $S_\Phi=[S_\Phi]''=\M{d}$. The result then follows from Proposition~\ref{prop:C0E=0}.
\end{proof}

We can now state and prove our main result.

\begin{proposition}\label{prop:bound for both}
    Let $\Phi: \M{d}\rightarrow \M{d}$ be a primitive unital quantum channel such that the associated operator systems $S_{\Phi^n}$ are C$^*$-algebras for all $n\in \mathbb{N}$. Then, 
    \begin{equation}
        q(\Phi)\leq d-2 \quad \text{and} \quad c(\Phi)=c_E(\Phi)\leq 2(d-1).
    \end{equation}
\end{proposition}

\begin{proof}
Note that for any channel $\Phi:\M{d}\to \M{d}$, we have an increasing chain of subspaces $S_\Phi \subseteq S_{\Phi^2}\subseteq \cdots \subseteq S_{\Phi^n}\subseteq \cdots$, which yields a decreasing chain of commutants
$[S_\Phi]' \supseteq [S_{\Phi^2}]'\supseteq \cdots \supseteq [S_{\Phi^n}]'\supseteq \ldots$. Note that the commutants $[S_{\Phi^n}]'$ are unital C$^*$-subalgebras of $\M{d}$ for all $n\in\mathbb{N}$ and by Lemma~\ref{lemma:chain suubalgebra}, we have
$[S_{\Phi^n}]'=\mathcal{M}_{\Phi^n}$. Since $\Phi$ is primitive, $\exists n\in\mathbb N$ such that $[S_{\Phi^n}]'=\mathbb{C}\iden$. So the above chain of C$^*$-algebras stabilzes at some $n$:
\begin{equation} 
[S_\Phi]' \supseteq [S_{\Phi^2}]'\supseteq \cdots \supseteq [S_{\Phi^n}]'=\mathbb{C}\iden. \label{eq:chain}
\end{equation}
Since $S_{\Phi^n}$ are algebras for all $n\in \mathbb N$, we can use the necessary and sufficient conditions given in Proposition \ref{prop:C0Q0=0} and Lemma~\ref{lemma:C0C0E=0} for the zero-error capacities to vanish. Firstly, note that Lemma~\ref{lemma:C0C0E=0} immediately tells us that $c(\Phi)=c_E(\Phi)$. Also note that if $[S_\Phi]'=\M{d}$, then by the double-commutant theorem, $S_\Phi=[S_\Phi]''=\mathbb{C}\iden$. It follows that the Choi rank of $\Phi$ is 1, which means that the channel is just a unitary conjugation. Such a channel can not be primitive, contradicting our hypothesis. Thus, we can assume that $[S_\Phi]'$ is a proper subalgebra of $\M{d}$. 

Now, if $[S_\Phi]'$ is trivial to begin with, then Lemma~\ref{lemma:C0C0E=0} shows that $C_0^{(1)}(\Phi)=0$. In this case, $c(\Phi)=1\leq 2(d-1)$. Hence, we can assume that $[S_\Phi]'$ is a proper non-trivial subalgebra of $\M{d}$. From Lemma \ref{lemma:chain suubalgebra}, we know that each containment in Eq.~\eqref{eq:chain} is proper. It is known that such a chain of decreasing unital C$^*$-subalgebras can have length at most $2(d-1)$ \cite[Lemma 5 and Theorem 3.6]{jaques-rahaman2017}. Thus $c(\Phi)\leq 2(d-1)$.

For the other bound, note that if $[S_\Phi]'$ is abelian to begin with, then $Q_0^{(1)}(\Phi)=0$ and $q(\Phi)=1\leq (d-2)$ (for $d\geq 2$), see Proposition~\ref{prop:C0Q0=0}. So, we can assume that $[S_\Phi]'$ is non-abelian. Now, in order to descend down the chain of commutants in Eq.~\eqref{eq:chain} all the way to $\mathbb{C}\iden$, in the worst case, there are $d$-many steps required from the full diagonal algebra to stabilize to the trivial algebra. Hence, it requires at most $2(d-1)-d=(d-2)$ steps to go from a non-abelian to an abelian algebra, proving that $q(\Phi)\leq d-2$.
\end{proof}

We now provide an example of a channel for which the operator system fulfils the requirements of the above proposition. 

\begin{example}
We construct a primitive unital channel $\Phi:\M{3}\to \M{3}$ for which $S_\Phi,S_{\Phi^2},$ and $S_{\Phi^3}$ are all algebras, $q(\Phi)=1$, and $c(\Phi)=c_E(\Phi)=3$. This is the example given after Theorem 3.9 in \cite{jaques-rahaman2017}, giving the maximum length of the chain in Eq.~\eqref{eq:chain} when the starting algebra is a {\em{maximal abelian subalgebra}} (MASA). 

The channel is defined by its Kraus operators $\{K_i\}_{i=1}^3 \subseteq \M{3}$, which are given below:

\[K_1=\frac{1}{\sqrt{2}}\begin{bmatrix}
    0 & 0 & 1\\
    0 & 0 & 1\\
    0 & 0 & 0
\end{bmatrix}, K_2= \frac{1}{\sqrt{2}}\begin{bmatrix}
    1 & 0 & 0\\
    -1 & 0 & 0\\
    0 & 0 & 0
\end{bmatrix}, K_3=\begin{bmatrix}
    0 & 0 & 0\\
    0 & 0 & 0\\
    0 & 1 & 0
\end{bmatrix}.\]
One can check that $S_\Phi=\Bigg\{\begin{bmatrix}
    a & 0 & 0\\
    0 & b & 0\\
    0 & 0 & c
\end{bmatrix}| \ a, b, c\in \mathbb{C}\Bigg\}.$ This is the full diagonal algebra. It is not hard to see that $S_{\Phi^2}= \Bigg\{\begin{bmatrix}
    a & 0 & b\\
    0 & c & 0\\
    d & 0 & e
\end{bmatrix}| \ a, b, c, d, e\in \mathbb{C}\Bigg\}\simeq \M{2}\oplus\M{1},$
where $\M{1}$ is just the scalar algebra. Hence $S_{\Phi^2}$ is also an algebra. And finally,  $S_{\Phi^3}=\M{3}.$ Thus, we obtain the chain $$S_{\Phi}\subset S_{\Phi^2}\subset S_{\Phi^3}=\M{3}.$$
Since $S_{\Phi^3}=\M{3}$, it follows that $[S_{\Phi^3}]'= \mathcal{M}_{\Phi^3} = \mathbb{C}\iden$ and hence $\Phi$ is primitive (see \cite[Corollary 3.5]{rahaman2017}). Furthermore, since $[S_{\Phi^2}]' \simeq \mathbb{C}\iden_2 \oplus \M{1}$ is not trivial,
it is clear from Lemma~\ref{lemma:C0C0E=0} that $c(\Phi)=c_E(\Phi)=3$. Finally, since $S_\Phi$ is the full diagonal algebra, $[S_\Phi]'=S_{\Phi}$ is also abelian and Proposition \ref{prop:C0Q0=0} shows that $q(\Phi)=1$. 
\end{example}
The example given above is a special case ($d=3$) for a more general construction given in \cite[Theorem 3.9]{jaques-rahaman2017}, where the decreasing chain of multiplicative domains have been created to provide the multiplicative index of the channel to be $d$. We think in the general case, the operator systems are also algebras and one can get  $c(\Phi)=d$, but we leave it as a future avenue to explore.

\subsection{Diagonal unitary covariant channels}\label{DUC-class}

In this section, we study the scrambling times and Wielandt indices of a special class of quantum channels that are covariant under the action of the diagonal unitary group. These channels were introduced and extensively studied in \cite{Singh2021diagonalunitary}. Here, we only recall some basic results. Note that we call $A\in \M{d}$ \emph{column stochastic} if it is entrywise non-negative and for all $i$, $\sum_j A_{ji}=1$. Furthermore, $\mathcal{DU}_d$ denotes the set of all diagonal unitary matrices in $\M{d}$.

\begin{theorem}\label{theorem:AB}
    For a channel $\Phi:\M{d}\to \M{d}$, the following are equivalent:
    \begin{itemize}
        \item $\forall X\in \M{d}, \forall U\in\mathcal{DU}_d : \,\, \Phi(UXU^*) = U \Phi(X) U^*$.
        \item $\exists A,B\in \M{d}$ with $A$ column stochastic and $B$ positive semi-definite such that 
        \begin{equation}
            \forall X\in \M{d}: \quad \Phi(X) = \sum_{i,j} A_{ij}X_{jj} \ketbra{i} + \sum_{i\neq j} B_{ij} X_{ij} \ketbra{i}{j} =: \Phi_{A,B}(X).
        \end{equation}
    \end{itemize}
    A channel $\Phi = \Phi_{A,B}$ as above is called \emph{conjugate diagonal unitary covariant} (CDUC).
\end{theorem}

\begin{theorem}
    For a channel $\Phi:\M{d}\to \M{d}$, the following are equivalent:
    \begin{itemize}
        \item $\forall X\in \M{d}, \forall U\in\mathcal{DU}_d : \,\, \Phi(UXU^*) = U^* \Phi(X) U$.
        \item $\exists A,C\in \M{d}$ with $A$ column stochastic, $C=C^{*}$, and $A_{ij}A_{ji}\geq \abs{C_{ij}}^2 \,\, \forall i,j$, such that 
        \begin{equation}
            \forall X\in \M{d}: \quad \Phi(X) = \sum_{i,j} A_{ij}X_{jj} \ketbra{i} + \sum_{i\neq j} C_{ij} X_{ji} \ketbra{i}{j} =: \Phi_{A,C}(X).
        \end{equation}
    \end{itemize}
    A channel $\Phi = \Phi_{A,C}$ as above is called \emph{diagonal unitary covariant} (DUC).
\end{theorem}

For the class of (C)DUC channels, we will show that the properties of strict positivity, scrambling, mixing, and primitivity are all equivalent to the corresponding properties of the classical stochastic matrix $A$. Let us first introduce the definitions of these properties for a stochastic matrix.

\begin{definition}
    A column stochastic matrix $A\in \M{d}$ is said to be
    \begin{itemize}
        \item \emph{strictly positive} if $A_{ij}>0 \,\, \forall i,j$.
        \item \emph{scrambling} if $\forall i,j, \,\, \exists k$ such that $A_{ki}A_{kj}>0$.
        \item \emph{mixing} if $\lambda=1$ is a simple eigenvalue of $A$ and $A$ has no other peripheral eigenvalues.
        \item \emph{primitive} if it is mixing and its unique invariant vector has full support.
    \end{itemize}
\end{definition}

\begin{remark}
    For a stochastic matrix, the spectral properties of mixing/primitivity can be verified by analyzing the connectivity of the directed graph associated with the matrix \cite[Section 2.3]{Singh2022ergodic}. 
\end{remark}

Note that if one uses classical channels of the form $\Phi_A := \Phi_{A,\textrm{diag} A}$, the quantum definitions of strict positivity, scrambling, mixing, and primitivity that were introduced in the previous sections all reduce to the classical definitions introduced above. Let us also observe that in the classical case, Theorem~\ref{theorem:scr-mix} and Corollary~\ref{corollary:primitive} reduce to the following results.

\begin{theorem}\label{theorem:Amixing}
    For a column stochastic $A\in \M{d}$, the following are equivalent
    \begin{itemize}
        \item $A$ is mixing.
        \item $\exists k\in \mathbb{N}$ such that $A^k$ is scrambling.
    \end{itemize}
\end{theorem}

\begin{theorem}\label{theorem:Aprimitive}
    For a column stochastic $A\in \M{d}$, the following are equivalent
    \begin{itemize}
        \item $A$ is primitive.
        \item $\exists k\in \mathbb{N}$ such that $A^k$ is strictly positive.
    \end{itemize}
\end{theorem}

We are now ready to prove some of our main results in this section.

\begin{theorem}\label{thm:strict-pos-DUC}
    The following equivalences hold for DUC and CDUC channels.
\begin{itemize}
    \item A \emph{CDUC} channel $\Phi_{A,B}:\M{d}\to \M{d}$ is strictly positive $\iff A$ is strictly positive.
    \item A \emph{DUC} channel $\Phi_{A,C}:\M{d}\to \M{d}$ is strictly positive $\iff A$ is strictly positive.
\end{itemize}
\end{theorem}
\begin{proof}
Clearly, if $\Phi_{A,B}$ or $\Phi_{A,C}$ is stricly positive, we can restrict to diagonal input states to conclude that $A$ is also stricly positive. Conversely, let $A$ be strictly positive. Let us first deal with the CDUC case. It suffices to show that $\Phi_{A,B}(\psi)$ is invertible for all pure states $\psi\in \St{d}$. If $\psi = \ketbra{\psi} = \ketbra{i}$ for some $i=0,1,\ldots, d-1$, invertibility of $\Phi_{A,B}(\psi)$ follows easily from strict positivity of $A$. Otherwise, there exist distinct $k\neq l$ such that $\psi_k \neq 0, \psi_l\neq 0$. Note that $\psi_i$ denotes the $i^{th}$ entry of the column vector $\ket{\psi}\in \C{d}$. In this case, for an arbitrary $\ket{\phi}\in \C{d}$ we can write
    \begin{align*}
    \operatorname{Tr}\big[\Phi_{A,B}(\psi)(\phi)\big]
    &= \sum_i\left(\sum_n A_{in}|\psi_n|^2 \right) \abs{\phi_i}^2 + \sum_{i\neq j} B_{ij} \psi_i\overbar{\psi_j\phi_i}\phi_j \\
    &= \sum_{i\neq n} A_{in} \abs{\psi_n}^2 \abs{\phi_i}^2 + \sum_{i} A_{ii}{\abs{\phi_i}^2} \abs{\psi_i}^2 + \sum_{i\neq j} B_{ij}\psi_i\overbar{\psi_j\phi_i}\phi_j \\
    &= \sum_{i\neq n} A_{in} \abs{\psi_n}^2 \abs{\phi_i}^2 + \left\langle \overbar{\psi}\odot \phi \big\vert \, B \, \big\vert \overbar{\psi}\odot \phi \right\rangle .
\end{align*}
Here, $\odot$ denotes the entrywise product of vectors and since $B$ is positive semi-definite, the second term above is always non-negative. Moreover, the first term is positive for all $\ket{\phi} \not\in \text{span}\{ \ket{k} \}$. For $\ket{\phi}=\ket{k}$, the first term is again positive since $\psi_l\neq 0$ and $l\neq k$. Thus, $\operatorname{Tr}\big[\Phi_{A,B}(\psi)(\phi)\big]>0$ for all $\ket{\psi},\ket{\phi}\in \C{d}$, which proves that $\Phi_{A,B}(\psi)$ is invertible for all pure $\psi\in \St{d}$.

In the DUC case, we proceed similarly. It suffices to show that $\Phi_{A,C}(\psi)$ is invertible for all pure states $\psi\in \St{d}$, when $A$ is strictly positive. If $\ket{\psi} = \ket{i}$ for some $i \in \{0,\ldots, d-1\}$, invertibility of $\Phi_{A,C}(\psi)$ follows easily from strict positivity of $A$. Otherwise, there exist distinct $k\neq l$ such that $\psi_k \neq 0, \psi_l\neq 0$. In this case, for an arbitrary $\ket{\phi}\in \C{d}$ we can write 
\begin{align*}
  \operatorname{Tr}\big[\Phi_{A,C}(\psi)\phi\big] &= \sum_i\left(\sum_n A_{in}|\psi_n|^2 \right) \abs{\phi_i}^2 + \sum_{i\neq j} C_{ij} \overbar{\psi_i}\psi_j \overbar{\phi_i}\phi_j \\
    &= \sum_i A_{ii} \abs{\psi_i}^2 \abs{\phi_i}^2 + \sum_{i\neq j} \left( A_{ij} \abs{\phi_i}^2 \abs{\psi_j}^2 + C_{ij} \overbar{\psi_i}\psi_j \overbar{\phi_i}\phi_j \right) \\ 
    &= \sum_i A_{ii} \abs{\psi_i}^2 \abs{\phi_i}^2 + \sum_{i< j} \begin{array}{c@{}@{}}
         \left(\begin{matrix}
                \psi_j \overbar{\phi_i} & \psi_i \overbar{\phi_j}
                \end{matrix}\right)  \\
                \hspace{0.01cm}
        \end{array}
        \left(\begin{matrix}
                    A_{ij} & C_{ij}  \\
                    C_{ji} & A_{ji}
                  \end{matrix}\right)
        \left(\begin{matrix}
                \overbar{\psi_j} \phi_i \\
                \overbar{\psi_i} \phi_j
                \end{matrix}\right)
\end{align*}
Note that all the terms inside the sum above are non-negative. Moreover, the first sum is positive for all $\ket{\phi}$ with either $\phi_k\neq 0$ or $\phi_l\neq 0$. If both $\phi_k=0$ and $\phi_l=0$ (i.e. $\ket{\phi}\perp \text{span}\{\ket{k},\ket{l} \}$, we can choose $p\neq k,l$ such that $\phi_p\neq 0$. Then, the $k,p$ block in the second sum above is positive, since
\begin{equation*}
  \begin{array}{c@{}@{}}
         \left(\begin{matrix}
                0 & \psi_k \overbar{\phi_p}
                \end{matrix}\right)  \\
                \hspace{0.01cm}
        \end{array}
        \left(\begin{matrix}
                    A_{kp} & C_{kp}  \\
                    C_{pk} & A_{pk}
                  \end{matrix}\right)
        \left(\begin{matrix}
                0 \\
                \overbar{\psi_k} \phi_p
                \end{matrix}\right) = A_{pk}\abs{\psi_k}^2 \abs{\phi_p}^2 > 0.
\end{equation*}
Thus, $\operatorname{Tr}\big[\Phi_{A,C}(\psi)(\phi)\big]>0$ for all $\ket{\psi},\ket{\phi}\in \C{d}$, which proves the desired result.
\end{proof}

\begin{theorem} \label{thm:scrambling-DUC} The following equivalences hold for DUC and CDUC channels.
\begin{itemize}
    \item A \emph{CDUC} channel $\Phi_{A,B}:\M{d}\to \M{d}$ is scrambling if and only if $A$ is scrambling.
    \item A \emph{DUC} channel $\Phi_{A,C}:\M{d}\to \M{d}$ is scrambling if and only if $A$ is scrambling.
\end{itemize}
\end{theorem}
\begin{proof}
In both the CDUC and the DUC case, our aim would be to appropriately decompose the function $f(\psi,\phi) = \operatorname{Tr}\big[\Phi(\psi)\Phi(\phi)\big]$ (for pure states $\psi,\phi\in\St{d}$) into non-negative parts so as to obtain the desired result. Let us tackle the CDUC case first. Here, we have 
\begin{align*}
    f_{A,B}(\psi,\phi) &= \operatorname{Tr}\big[\Phi_{A,B}(\psi)\Phi_{A,B}(\phi)\big] \\
    &= \sum_{i,j=1}^d \Phi_{A,B}(\psi)_{ij} \overbar{\Phi_{A,B}(\phi)_{ij}} \\
    &= \sum_i\left(\sum_k A_{ik}|\psi_k|^2 \right) \left(\sum_l A_{il}|\phi_l|^2 \right) + \sum_{i\neq j} |B_{ij}|^2 \psi_i\overbar{\psi_j\phi_i}\phi_j \\
    &= \sum_i \sum_{k\neq l} A_{ik}A_{il}|\psi_k|^2 |\phi_l|^2 + \sum_{i,k} A_{ik}^2 |\psi_k|^2 |\phi_k|^2 + \sum_{i\neq j} |B_{ij}|^2 \psi_i\overbar{\psi_j\phi_i}\phi_j \\
    &= \sum_{k\neq l} (A^\top A)_{kl} |\psi_k|^2 |\phi_l|^2 + \sum_{i\neq k} A_{ik}^2 |\psi_k|^2 |\phi_k|^2 + \sum_i A_{ii}^2 |\psi_i|^2 |\phi_i|^2 +  \sum_{i\neq j} |B_{ij}|^2 \psi_i\overbar{\psi_j\phi_i}\phi_j \\
    &= \sum_{k\neq l} (A^\top A)_{kl} |\psi_k|^2 |\phi_l|^2 + \sum_{i\neq k} A_{ik}^2 |\psi_k|^2 |\phi_k|^2 + \left\langle \overbar{\psi}\odot \phi \big\vert \, B\odot \overbar{B} \, \big\vert \overbar{\psi}\odot \phi \right\rangle
\end{align*}
Notice that since $A$ is entrywise non-negative and $B$ is positive semi-definite, all three terms above are non-negative. Now, assume that $\Phi_{A,B}$ is scrambling, so that $f_{A,B}(\psi,\phi)>0$ for all pure states $\psi, \phi$. Then, we can choose $\ket{\psi}=\ket{k}$ and $\ket{\phi}=\ket{l}$ for $k\neq l$, so that 
\begin{equation*}
    f_{A,B}(\psi,\phi) = (A^\top A)_{kl} > 0 \implies A \text{ is scrambling}.
\end{equation*}
Conversely, if $A$ is scrambling, i.e., $(A^\top A)_{kl} > 0$ for all $k\neq l$, then for any two orthogonal pure states $\psi\perp \phi$, by identifying indices $k\neq l$ such that $\psi_k\neq 0$ and $\phi_l\neq 0$, we get 
\begin{equation*}
    f_{A,B}(\psi,\phi) \geq (A^\top A)_{kl} |\psi_k|^2 |\phi_l|^2 > 0 \implies \Phi_{A,B} \text{ is scrambling}.
\end{equation*}

Now, for a DUC channel $\Phi_{A,C}$, we can follow the same steps as above to obtain
\begin{align*}
    f_{A,C}(\psi,\phi) &= \operatorname{Tr}\big[\Phi_{A,C}(\psi)\Phi_{A,C}(\phi)\big] \\
    &= \sum_{i,j=1}^d \Phi_{A,C}(\psi)_{ij} \overbar{\Phi_{A,C}(\phi)_{ij}} \\
    &= \sum_i\left(\sum_k A_{ik}|\psi_k|^2 \right) \left(\sum_l A_{il}|\phi_l|^2 \right) + \sum_{i\neq j} |C_{ij}|^2 \overbar{\psi_i}\psi_j\phi_i\overbar{\phi_j} \\
    &= \sum_i \sum_{k\neq l} A_{ik}A_{il}|\psi_k|^2 |\phi_l|^2 + \sum_{i,k} A_{ik}^2 |\psi_k|^2 |\phi_k|^2 + \sum_{i\neq j} |C_{ij}|^2 \overbar{\psi_i}\psi_j\phi_i\overbar{\phi_j} \\
    &= \sum_{k\neq l} (A^\top A)_{kl} |\psi_k|^2 |\phi_l|^2 + \sum_i A_{ii}^2 |\psi_i|^2 |\phi_i|^2 + \sum_{i\neq j} A_{ij}^2 |\psi_j|^2 |\phi_j|^2 + \sum_{i\neq j} |C_{ij}|^2 \overbar{\psi_i}\psi_j\phi_i\overbar{\phi_j} \\
    &= \sum_{k\neq l} (A^\top A)_{kl} |\psi_k|^2 |\phi_l|^2 + \sum_i A_{ii}^2 |\psi_i|^2 |\phi_i|^2 + \sum_{i<j} \begin{array}{c@{}@{}}
         \left(\begin{matrix}
                \psi_j \overbar{\phi_j} & \psi_i \overbar{\phi_i}
                \end{matrix}\right)  \\
                \hspace{0.01cm}
        \end{array}
        \left(\begin{matrix}
                    A_{ij}^2 & |C_{ij}|^2  \\
                    |C_{ji}|^2 & A_{ji}^2
                  \end{matrix}\right)
        \left(\begin{matrix}
                \overbar{\psi_j} \phi_j \\
                \overbar{\psi_i} \phi_i
                \end{matrix}\right)
\end{align*}
As before, since $\Phi_{A,C}$ is a channel, the constraints on $A,C$ force all three sums above to be non-negative. The remaining argument is an exact replica of the one used in the CDUC case. 
\end{proof}

\begin{theorem} \label{thm:DUC-primitive}
The following equivalences hold for DUC and CDUC channels.
\begin{itemize}
    \item A \emph{CDUC} channel $\Phi_{A,B}:\M{d}\to \M{d}$ is primitive if and only if $A$ is primitive.
    \item A \emph{DUC} channel $\Phi_{A,C}:\M{d}\to \M{d}$ is primitive if and only if $A$ is primitive.
\end{itemize}
\end{theorem}
\begin{proof}
    We only tackle the CDUC case here. Using Corollary~\ref{corollary:primitive} and Theorems~\ref{theorem:Aprimitive}, \ref{thm:strict-pos-DUC}, we obtain
    \begin{align}
        \Phi_{A,B} \text{ is primitive } &\iff \exists k\in \mathbb{N} \text{ such that } \Phi^k_{A,B} = \Phi_{A^k, f(B)} \text{ is strictly positive } \\ 
        &\iff \exists k\in \mathbb{N} \text{ such that } A^k \text{ is strictly positive} \iff A \text{ is primitive}.
    \end{align}
Note that above, $f(B)=B^{\odot k} + \operatorname{diag}(A^k - A^{\odot k})$. This follows from the composition rule:
\begin{equation}
    \Phi_{A_1,B_1}\circ \Phi_{A_2,B_2} = \Phi_{A_1 A_2, B_1 \odot B_2 + \operatorname{diag}(A_1 A_2 - A_1 \odot A_2)}.
\end{equation}
Note that $\odot$ here denotes the entrywise (or Hadamard) product of matrices.
\end{proof}

\begin{theorem} \label{thm:DUC-mixing}
The following equivalences hold for DUC and CDUC channels.
\begin{itemize}
    \item A \emph{CDUC} channel $\Phi_{A,B}:\M{d}\to \M{d}$ is mixing if and only if $A$ is mixing.
    \item A \emph{DUC} channel $\Phi_{A,C}:\M{d}\to \M{d}$ is mixing if and only if $A$ is mixing.
\end{itemize}
\end{theorem}
\begin{proof}
    We only tackle the CDUC case here. Using Theorems~\ref{theorem:scr-mix}, \ref{theorem:Amixing} and \ref{thm:scrambling-DUC}, we obtain
    \begin{align}
        \Phi_{A,B} \text{ is mixing } &\iff \exists k\in \mathbb{N} \text{ such that } \Phi^k_{A,B}  = \Phi_{A^k, f(B)} \text{ is scrambling } \\ 
        &\iff \exists k\in \mathbb{N} \text{ such that } A^k \text{ is scrambling} \iff A \text{ is mixing}.
    \end{align}
\end{proof}

We now shift our focus to the scrambling times and Wielandt indices of (C)DUC channels. We will borrow results from the classical literature to provide optimal upper bounds on these indices for (C)DUC channels. Let us first define these indices for stochastic matrices. 

\begin{definition}
    For a mixing (resp. primitive) stochastic matrix $A\in \M{d}$, we define
    \begin{equation}        c(A):= \min \{n : A^n \,\, \text{is scrambling} \} \quad \text{resp.} \quad w(A):= \min \{n : A^n \,\, \text{is strictly positive} \}.
    \end{equation}
\end{definition}
We call $c(A)$ the \emph{time} of $A$ and $w(A)$ the \emph{Wielandt index}\footnote{The number $w(A)$ is also sometimes called the {\em{primitivity index}} of $A$ (see e.g.~\cite{wielandt2010}).} of $A$.
Optimal bounds are known for these indices, which are stated below. The following stochastic matrix (in appropriate dimension $d$) serves to prove the optimality of these bounds:
\begin{equation}\label{eq:Ad2}
    A_d := \left( \begin{array}{cccccc}
        0 & 1/2 & 0 & 0 & 0 & 0 \\
        0 & 0 & 1 & 0 & 0 & 0 \\
        0 & 0 & 0 & 1 & 0 & 0 \\
        \vdots & \vdots & \vdots & \vdots & \ddots & \vdots \\
        0 & 0 & 0 & 0 & 0 & 1 \\
        1 & 1/2 & 0 & 0 & 0 & 0 \\
    \end{array} \right)
\end{equation}

\begin{theorem} \cite{Wielandt1950index}\label{theorem:classical-wie-index}
    For any primitive stochastic matrix $A\in \M{d}$:
    \begin{equation} 
    w(A) \leq d^2 -2d +2.
    \end{equation}
    Moreover, $w(A_d)=d^2 - 2d +2$.
    \label{thm:6.7}
\end{theorem}

\begin{theorem} \cite{Akelbek2009index, Guterman2019index} \label{theorem:classical-scr-index}
    For any mixing stochastic matrix $A\in \M{d}$:
    \begin{equation} 
    c(A) \leq \ceil[\bigg]{\frac{d^2 -2d +2}{2}}.
    \end{equation} 
    Moreover, $c(A_d)= \ceil[\bigg]{\frac{d^2 -2d +2}{2}}$.
      \label{thm:6.8}
\end{theorem}

Using Theorems~\ref{thm:strict-pos-DUC}-\ref{thm:DUC-mixing}, \ref{theorem:classical-wie-index}, and \ref{theorem:classical-scr-index}, the corollaries given below follow immediately. 

\begin{corollary}\label{corollary:DUC-mixing}
    Let $A\in \M{d}$ be a mixing stochastic matrix. Then, for any CDUC channel $\Phi_{A,B}:\M{d}\to \M{d}$ and any DUC channel $\Phi_{A,C}:\M{d}\to \M{d}$,
    \begin{equation}
        c(\Phi_{A,B})=c(\Phi_{A,C})=c(A)\leq d^2 -2d +2.
    \end{equation}
    Equality is achieved above for $A=A_d$ (see Eq.~\eqref{eq:Ad2}).
\end{corollary}
\begin{corollary}\label{corollary:DUC-primitive}
     Let $A\in \M{d}$ be a primitive stochastic matrix. Then, for any CDUC channel $\Phi_{A,B}:\M{d}\to \M{d}$ and any DUC channel $\Phi_{A,C}:\M{d}\to \M{d}$,
    \begin{equation}
        w(\Phi_{A,B})=w(\Phi_{A,C})= w(A)\leq \ceil[\bigg]{\frac{d^2 -2d +2}{2}}.
    \end{equation}
   Equality is achieved above for $A=A_d$ (see Eq.~\eqref{eq:Ad2}). 
\end{corollary}

\subsection{Trade-off relation}
In this short section, we connect the one-shot zero-error capacities of a channel with that of its complementary channel. Recall that for any channel $\Phi: \M{d}\rightarrow\M{d}$ with linearly independent Kraus representation
\[\Phi(X)=\sum_{i=1}^p K_i X K_i^*,\]
there is a complementary channel $\Phi^{C}:\M{d}\rightarrow \M{p}$ defined by

\[\Phi^{C}(X)=\sum_{i,j} \Tr(K_i^*K_j X) E_{i,j},\]
where $\{E_{i,j}\}$ are the matrix units of $\M{p}$. 

Note that if $\Phi^{C^*}$ denotes the adjoint of the channel $\Phi^{C}$, then %one calculates
\[\Tr(K_i^*K_j X)=\Tr(E_{j,i} \Phi^{C}(X))=\Tr(\Phi^{C^*}(E_{i,j}^*)X),\]
for all $X\in \M{d}$ and $1\leq i,j\leq p$. Thus, it follows that $\Phi^{C^*}(E_{i,j})=K_j^*K_i$ and hence we get the operator system of $\Phi$ as the image of $\Phi^{C^*}$, that is,
\[S_\Phi=\text{span} \{K_i^*K_j: 1\leq i, j\leq p\}=\text{range} (\Phi^{C^*}).\]
In the following proposition we provide a relation between the one-shot zero-error quantum capacity of $\Phi$ to the one-shot zero-error classical capacity of its complementary channel $\Phi^C$.
\begin{proposition}\label{trade-off}
    Let $\Phi: \M{d}\rightarrow\M{d}$ be a quantum channel and let $\Phi^{C}:\M{d}\rightarrow \M{p}$ be its complementary channel. Then, it holds that
      
    $$2^{C^{(1)}_0(\Phi^C)} + 2^{Q^{(1)}_0 (\Phi)} \leq d+1.$$

%\end{enumerate}

\end{proposition}
\begin{proof} The basic idea of the proof is the complementary relation between the error correcting subspaces of a channel and the private subspaces of its complementary channel (see \cite{KKS, KSW}). 

    Let $\mathcal{C}\subseteq \C{d}$ be the largest subspace in which $\Phi$ can be recovered, i.e., $Q^{(1)}_0(\Phi) = \log \dim \mathcal{C}$. More precisely, $\mathcal{C}$ is the largest subspace where $\Phi$ admits a channel $\mathcal{R}$ such that $\mathcal{R}\circ\Phi(\rho)=\rho$ for all states $\rho\in \State{\C{d}}$ with $\supp (\rho)\subseteq \mathcal{C}$. From the Knill-Laflamme condition for error-correction \cite{knil-laf}, it holds that 
    \[ \forall i,j: \quad P_{\mathcal{C}} K_i^*K_jP_{\mathcal{C}}=\lambda_{i,j}P_{\mathcal{C}},\]
where $P_{\mathcal{C}}$ is the projection onto $\mathcal{C}$, $\{K_i\}_i$ are the Kraus operator of $\Phi$, and $\lambda_{i,j}\in \mathbb{C}$. Note that any $\rho$ supported on $\mathcal{C}$ satisfies $\rho=P_{\mathcal{C}} \rho P_{\mathcal{C}}$. Thus, for any such $\rho$, we obtain 

\[\Phi^C(\rho)=\sum_{i,j} \Tr(K_i^*K_j P_{\mathcal{C}} \rho P_{\mathcal{C}}) E_{i,j}=\sum_{i,j} \Tr(P_{\mathcal{C}} K_i^*K_j P_{\mathcal{C}} \rho) E_{i,j}=\sum_{i,j}\lambda_{i,j} \Tr( P_{\mathcal{C}} \rho)E_{i,j} = \Tr (\rho) X,\]
where $X=\sum_{i,j}\lambda_{i,j} E_{i,j}$ is a positive semi-definite matrix. Therefore, in any zero-error encoding $\{\ket{\psi_i} \}_{i=1}^M$ of classical messages $\{1,2,\ldots ,M \}$ to be sent through $\Phi^C$, there can only be at most one code state from $\mathcal{C}$, which means that $M \leq (d - \dim \mathcal{C}) + 1$. Optimizing over all zero-error classical encodings gives us the required bound:
\begin{equation}
    2^{C^{(1)}_0(\Phi^C)} + 2^{Q^{(1)}_0 (\Phi)} \leq d+1.
\end{equation}
\end{proof}

\end{appendices}

\end{document}